\documentclass[journal]{IEEEtran}
\usepackage{graphicx,amssymb,amsmath}
\usepackage[noadjust]{cite}
\usepackage{setspace}               
\usepackage{stfloats}

\usepackage[normal]{threeparttable}
\usepackage{amsthm}
\usepackage{amsmath}
\usepackage{bbm}
\usepackage{flushend}
\usepackage{cases,subeqnarray}
\usepackage{bm,multirow,bigstrut}
\usepackage{textcomp}
\usepackage{latexsym,bm}
\usepackage{booktabs,changebar}
\usepackage{xcolor}
\usepackage{mathtools}
\usepackage{dsfont}
\usepackage{extarrows}
\usepackage{mathrsfs}
\usepackage{cite}
\usepackage{bm}
\usepackage{cleveref}
\usepackage{multicol}       
\usepackage{multirow}       
\usepackage{array}          
\usepackage{colortbl}
\usepackage{makecell}
\usepackage{xcolor}
\definecolor{crimson}{RGB}{192,0,0}         
\definecolor{navy}{RGB}{47,85,151}         

\makeatletter                               
\newif\if@restonecol
\makeatother

\makeatletter
\newif\if@restonecol
\makeatother

\usepackage[linesnumbered,ruled,vlined]{algorithm2e}
\usepackage{algpseudocode}
\usepackage{amsmath}
\renewcommand{\arraystretch}{1.5} %

\theoremstyle{plain}
\newtheorem{thm}{Theorem}
\newtheorem{lemm}{Lemma}

\newtheorem{coro}{Corollary}

\theoremstyle{plain}

\def\diag{\mathrm{diag}}

\def\Htran{\mbox{\tiny $\mathrm{H}$}}

\begin{document}

\title{Channel Map-Based Channel Estimation for Near-Field UM-MIMO with Movable Planar Arrays\vspace{-0.em}}

\author{
Shuaifei~Chen,~\IEEEmembership{Member,~IEEE}, Cheng-Xiang~Wang,~\IEEEmembership{Fellow,~IEEE}, Chen~Huang,~\IEEEmembership{Member,~IEEE}, \\Qianze~Yang, Xiping~Wu,~\IEEEmembership{Senior Member,~IEEE}, Yunfei~Chen, \IEEEmembership{Fellow,~IEEE}, \\and El-Hadi~M.~Aggoune,~\IEEEmembership{Life Senior Member, IEEE}
\vspace{-2em}
\thanks{This work was supported by the National Natural Science Foundation of China (NSFC) under Grants 62394290, 62394291, and 62401643, the Major Science and Technology Project of Jiangsu Province under Grant BG2025039, the Young Elite Scientists Sponsorship Program by China Association for Science and Technology, under Grant 2022QNRC001, the Research Fund of National Mobile Communications Research Laboratory, Southeast University, under Grant 2025A05, and the Promising Researcher Program, University of Tabu, Saudi Arabia, under Grant PRP 2025-01.
Part of this paper was accepted to be presented at the IEEE Global Communications Conference 2026 (GC 2026)~\cite{yang2026improving}. (Corresponding authors: Cheng-Xiang Wang and Chen Huang.)

S. Chen and C. Huang are with Purple Mountain Laboratories, Nanjing, 211111 China; and with the National Mobile Communications Research Laboratory, School of Information Science and Engineering, Southeast University, Nanjing 211189, China (e-mail: shuaifeichen@seu.edu.cn, huangchen@pmlabs.com.cn).

C.-X. Wang and X. Wu are with the National Mobile Communications Research Laboratory, School of Information Science and Engineering, Southeast University, Nanjing 211189, China; and with the Pervasive Communication Research Center, Purple Mountain Laboratories, Nanjing 211111, China (e-mail: chxwang@seu.edu.cn, xiping.wu@seu.edu.cn).

Q. Yang is with the National Mobile Communications Research Laboratory, School of Information Science and Engineering, Southeast University, Nanjing 211189, China (e-mail: qianz\_yang@seu.edu.cn).

Y. Chen is with the Department of Engineering, University of Durham, DH1 3LE Durham, U.K. (e-mail: yunfei.chen@durham.ac.uk).

E.-H. M.~Aggoune is with the AI and Sensing Technologies Research Center, University of Tabuk, Tabuk 47315, Saudi Arabia (e-mail: haggoune@ut.edu.sa).
}
}


\maketitle
\begin{abstract}
  Accurate channel estimation is essential for coherent transmission in ultra-massive multiple-input multiple-output (UM-MIMO) systems, where near-field propagation and high-dimensional spatial channels impose substantial signal processing challenges. Movable antenna architectures increase the estimation complexity further due to geometry-dependent channel variations. Existing approaches struggle to balance accuracy and complexity, motivating the use of environment-dependent propagation structures for efficient estimation. To this end, this paper proposes a unified channel map-based channel estimation framework for UM-MIMO systems, which integrates movable planar array reconfiguration and near-field spherical-wave modeling to support geometry-aware line-of-sight (LoS) estimation and efficient non-LoS (NLoS) recovery. A channel map-based LoS estimator is proposed combining coarse user position information with a Fisher information-guided antenna placement strategy. Two efficient NLoS estimation methods are also presented, including a sketch-based reduced-subspace estimator for low-complexity processing and a channel map-based estimator that leverages scatterer location information for near-optimal performance. The framework further incorporates visibility-region modeling and a structural similarity-based pilot assignment strategy for multi-user scenarios. Simulation results show that the proposed channel map-based framework improves estimation accuracy, reduces computational overhead, and enhances scalability compared with state-of-the-art benchmarks without channel maps. 
\end{abstract}

\section{Introduction}

\IEEEPARstart{T}{he} sixth-generation (6G) wireless networks are envisioned to usher in a new era of global connectivity with ambitious performance gains over the fifth-generation (5G)~\cite{wang2023road}.
6G, comprising an order-of-magnitude increase in access points (APs), devices, and antennas, is expected to deliver tenfold to hundredfold improvements in metrics such as spectral efficiency (SE) and energy efficiency.
These gains will be enabled by integrating a range of key technologies, such as cell-free massive multiple-input multiple-output (M-MIMO)~\cite{demir2022cell,chen2021structured,wang2026sparse}, ultra-M-MIMO (UM-MIMO)~\cite{lu2024tutorial}, millimeter-wave/terahertz bands~\cite{wang2021general}, advanced signal processing \cite{wang2025optimal}, and electromagnetic information theory-based system design~\cite{wang2025modeling}.
In particular, UM-MIMO has emerged as a cornerstone of 6G, where each node may be equipped with extremely large antenna arrays to boost capacity.
Unlike conventional M-MIMO typically operating in the far field, UM-MIMO deployments will frequently operate in the near-field region due to their large aperture dimensions and high frequencies.
As the Fraunhofer (Rayleigh) distance grows quadratically with antenna size, the plane-wave assumption of far-field propagation breaks down, and the spherical wavefronts of near-field propagation must be accounted for in channel models and signal design.
This transition to near-field propagation brings new opportunities.
For example, beamforming can be beam-focusing rather than broad beam steering, concentrating energy on the specific range and angle of a user equipment (UE) to mitigate the beam split effect~\cite{cui2024near}.
At the same time, it complicates channel estimation through distance-dependent array responses~\cite{long2026channel}.

Movable antenna architectures have been proposed to further enhance the spatial degrees of freedom in UM-MIMO systems \cite{zhu2025tutorial,wong2020fluid}.
In particular, movable planar arrays (MPAs), which represent a type of programmable metasurface-enabled architectures~\cite{tang2020wireless}, allow physical repositioning of the antenna array within a certain area or trajectory to improve coverage and signal strength.
By capitalizing on the spatial variability of the wireless channels, such dynamic array movement can help mitigate interference and adapt to multipath fading, thereby improving SE and link reliability~\cite{ma2023mimo}.
However, harnessing MPA gains requires rethinking channel modeling and estimation, since the channel responses become explicitly dependent on the array's position. 
Each perturbation of the antenna position alters the propagation geometry, rendering the channel a function of the locations of both UE and MPA.
This leads to considerable complexity in obtaining accurate channel state information (CSI), as traditional static-array estimation techniques may not directly apply to MPAs~\cite{zhu2025tutorial}.
The focus of this work is on low-complexity near-field channel estimation in such scenarios, addressing the intertwined challenges of spherical-wave propagation and array antenna movement.

Researchers have recently made progress in near-field channel estimation techniques tailored for UM-MIMO.
A distinctive feature of the near-field regime is that a propagation channel can be viewed as a superposition of a strong line-of-sight (LoS) component plus numerous non-LoS (NLoS) scattered components~\cite{yuan2024scalable,kang2025pilot}.
By first estimating the LoS component using parameters such as angle and distance, instead of treating it as deterministic, and then subtracting it, one obtains a residual measurement for the multipath NLoS component and addresses it with appropriate sparsity or statistical models~\cite{garkisch2024user}.
Nonetheless, a fully optimal estimation minimizing the mean-squared error (MSE) would require knowledge of the channel covariance, which is impractical for thousands of antennas~\cite{kay1993fundamentals}.
In fact, acquiring the massive spatial correlation matrix and implementing a minimum MSE (MMSE) estimator is computationally intractable for UM-MIMO dimensions.
This has motivated the research for estimation methods that exploit channel properties without heavy prior information.

One promising class of techniques is subspace-based estimation with reduced dimensions~\cite{kosasih2025spatial}.
These methods assume that the high-dimensional near-field channel lies predominantly in a lower-dimensional subspace determined by the array geometry or angular scattering characteristics, in which the bulk of the channel energy resides.
Notably, the reduced-subspace least squares (RSLS) approach constructs an approximate basis spanning the relevant angle and distance domain for the array, and then projects the received observations onto this basis for channel estimation.
By focusing the estimation within this smaller subspace, RSLS can dramatically reduce complexity and mitigate noise amplification, thereby achieving performance close to the MMSE bound without requiring the covariance matrix.
Recent studies have developed RSLS estimators for uniform planar arrays (UPAs) that incorporate the distance-dependent array responses into the basis construction~\cite{demir2024spatial}.
These methods show clear benefits over a naive full-dimensional least squares (LS) estimator, which ignores spatial correlation.
However, RSLS performance depends on selecting an appropriate subspace: an overly broad angular/radial support retains high dimensionality, whereas a narrow or mismatched support may introduce bias.
Thus, subspace methods typically require propagation-environment priors, such as user-range or cluster-location information.

Indeed, a recent trend in 6G research is to exploit environment-aware channel knowledge to aid communication tasks~\cite{zeng2024tutorial}.
Instead of treating the channel as a purely random entity to be estimated anew in each coherence block, the idea is to build a channel map of the radio environment that catalogs large-scale channel information for different locations.
A channel map is essentially a site-specific database stored at the base stations (BSs) or central processing units that takes the UE's position as input to yield prior information regarding the channels in that vicinity~\cite{qi2025novel}.
This information can include expected path loss or shadowing, angular power distribution, dominant scatterer locations, or even predicted instantaneous channel coefficients.
By leveraging the channel maps, the system operates in an environment-aware manner, reducing real-time pilot overhead and improving CSI accuracy.
{As long-term or quasi-static priors, these maps are periodically updated from sensing or channel measurements on the environmental timescale, with earlier refreshes triggered by persistent inconsistency with current observations~\cite{zeng2024tutorial}.}
This paradigm shifts conventional environment-unaware communications to a proactive approach, enhancing tasks like beamforming~\cite{wu2023environment}, multiple access~\cite{chen2025channel}, network optimization~\cite{zhan2023aerial}, etc.
In the context of channel estimation for UM-MIMO, one could utilize channel map-provided structural information, such as known dominant scatterer angles or LoS blockage, to guide the estimation process toward the correct channel parameters.
This can improve both accuracy and efficiency.
However, existing studies on channel maps have primarily focused on static antenna deployments, and the integration of channel maps with MPAs has yet to be fully explored.

In summary, critical gaps remain in the state-of-the-art.
First, there is a lack of scalable near-field channel estimation techniques practical for UM-MIMO.
Existing methods either require high-dimensional prior knowledge or suffer from degraded performance when that knowledge is imperfect. Second, the concept of channel map has not yet been fully integrated with physical-array adaptation by jointly exploiting environment-aware CSI and MPA-driven geometry control for channel estimation.
Third, multi-UE interference, e.g., pilot contamination, in near-field UM-MIMO could be tackled by leveraging environmental information or the unique spatial structure of near-field channels.
To this end, in this paper, we propose a unified channel map-based near-field channel estimation framework for UM-MIMO with MPAs.
Our main contributions are summarized as follows:
\begin{itemize}
  \item We propose a channel map aided LoS estimation scheme that leverages prior knowledge of coarse UE positions from a UE distribution map (UDM) together with a movable antenna optimization strategy guided by a {Fisher information-based lower-bound (FILB) benchmark}.
      This design improves geometric diversity and enables highly accurate LoS estimation under near-field spherical wave propagation scenarios.
  \item We propose two efficient NLoS estimation methods that exploit the low-rank scattering structure of near-field channels.
      The first employs a sketch algorithm-based RSLS (SA-RSLS) estimator for low-complexity subspace acquisition. The second adopts a channel map-based RSLS (CM-RSLS) estimator that incorporates scatterer information from a scatterer distribution map (SDM) and approaches genie-aided accuracy with only a compact QR decomposition (QRD).
  \item We extend the proposed NLoS estimation framework to multi-UE scenarios with pilot contamination by modeling scatterer visibility regions and constructing a multi-UE channel map-aided estimator.
      A structural similarity-based pilot assignment strategy is further developed to reduce subspace overlap among UEs, enhance robustness under limited pilot resources, and improve SE.
  \item We provide comprehensive performance evaluations demonstrating that the proposed schemes improve LoS and NLoS estimation accuracy, mitigate pilot contamination, and reduce computational complexity.
      These results confirm that incorporating environment-aware channel maps into near-field UM-MIMO enables a principled and scalable pathway toward highly accurate channel estimation and improved uplink transmission performance.
\end{itemize}

The rest of this paper is organized as follows.
Section~\ref{sec:system} introduces the considered system model.
{Section~\ref{sec:LoS estimation} and Section~\ref{sec:NLoS estimation} detail the proposed channel map-based estimation schemes for the LoS and NLoS components, respectively.}
Section~\ref{sec:NLoS estimation MU} extends the proposed NLoS estimation framework to multi-UE scenarios.
The results and discussion are provided in Section~\ref{sec:results}.
Finally, Section~\ref{sec:conclusion} draws the conclusions.

{\it Notation}: The boldface lowercase letters, $\bf x$, boldface uppercase letters, $\bf X$, and calligraphic uppercase letters, $\cal A$, denote the column vectors, matrices, and sets, respectively.
We denote by $[{\bf x}]_i$ the $i$-th element of a vector ${\bf x}$ and $[{\bf X}]_{ij}$ the $(i,j)$-th element of a matrix ${\bf X}$.
{We denote by $\diag(x_1,\ldots,x_n)$ a diagonal matrix with $x_1,\ldots,x_n$ on the diagonal and ${\rm blkdiag}({\bf X}_1,\ldots,{\bf X}_n)$ a block-diagonal matrix with square matrices ${\bf X}_1,\ldots,{\bf X}_n$ on the diagonal.}
{We use ${\rm tr}({\bf X})$ for the trace of a square matrix $\bf X$.}
The superscripts $(\cdot)^{\top}$ and $(\cdot)^{\Htran}$ denote the transpose and conjugate transpose, respectively.
We denote by ${\cal N}_{\mathbb C}\left({{\bf 0},{\bf R}}\right)$ the multivariate circularly symmetric complex Gaussian distribution with correlation matrix $\bf R$.
We denote by ${\bf I }_n$ the $n \!\times \! n$ identity matrix and $\|\cdot\|$ the $\ell_2$-norm.
We use ${\mathbb E}\{ \cdot \}$ to compute the expectation values.

\section{System Model}\label{sec:system}

Consider a UM-MIMO system where the BS is equipped with a movable planar array (MPA) comprising $N$ movable antennas.
As illustrated in Fig.~\ref{fig:1_system}, the origin of the three-dimensional (3D) Cartesian coordinate system is located at the corner of the array and the position of the $n$-th antenna is given by ${\bf a}_n = [0,y_n,z_n]^\top \in {\cal Q}$, where ${\cal Q}$ denotes the region within which the movable antennas can move freely, $n=1,\ldots,N$.
To avoid mutual coupling, a minimum spacing $d$ is enforced between any two movable antennas, i.e., $\|{\bf a}_m-{\bf a}_n\| \ge d$ for $m\neq n$.
We consider a near-field communication scenario.
More precisely, {an arbitrary UE with position ${\bf q} \!=\! [x_{\rm u},y_{\rm u},z_{\rm u}]^\top$} and $L$ scatterers with positions ${\bf p}_l = [x_l,y_l,z_l]^\top $, $l=1,\ldots,L$, are distributed within the range between the Fraunhofer (Rayleigh) distance $r_{\rm FA} = {2D^2}/{\lambda}$ and the Fresnel distance $r_{\rm FS} = 0.62 \sqrt{{D^3}/{\lambda}}$, where $D = \max_{1\le n\le N} \| {\bf a}_n - \frac{1}{N}\sum_{i=1}^{N} {\bf a}_i \|$ and $\lambda >0$ represent the array aperture size and carrier wavelength, respectively~\cite{zhu2025tutorial,wong2020fluid}.
We consider the millimeter-wave band with carrier frequency $f = 28$ GHz and wavelength $\lambda = 3\times 10^8/f \approx 10.7$ mm.

{Consider an arbitrary point ${\bf p}=r{\bf u}(\varphi,\theta)$, where $r=\|{\bf p}\|$ is the distance from ${\bf p}$ to the origin, ${\bf u}(\varphi,\theta)=[\cos\theta\cos\varphi,\cos\theta\sin\varphi,\sin\theta]^\top$ is the radial unit vector, and $\varphi$ and $\theta$ denote the azimuth angle-of-arrival (AoA) and the elevation AoA, respectively.
Here, ${\bf p}$ denotes a generic propagation point.
The array response vector for the considered MPA with respect to the origin is given by
\vspace{-0.5em}\begin{equation}\label{eq:b}
{\bf b}(\varphi,\theta,r) = [e^{-\jmath \chi(d_{1} - r)},\ldots, e^{-\jmath \chi(d_{N} - r)}]^\top \in {\mathbb C}^{N}
\vspace{-0.25em}\end{equation}
where $\chi = {2\pi}/{\lambda}$ denotes the wave number and $d_n=\|{\bf p}-{\bf a}_n\|$ is the distance from ${\bf p}$ to the $n$-th antenna of the MPA.
More precisely, we have
\vspace{-0.5em}\begin{align}\label{eq:dn}
&d_n=\|r{\bf u}(\varphi,\theta)-{\bf a}_n\| = \sqrt{r^2-2r{\bf u}^\top(\varphi,\theta){\bf a}_n+\|{\bf a}_n\|^2} \notag\\
&\approx r-{\bf u}^\top(\varphi,\theta){\bf a}_n
+\frac{\|{\bf a}_n\|^2-[{\bf u}^\top(\varphi,\theta){\bf a}_n]^2}{2r}.
\vspace{-0.25em}\end{align}
This second-order expansion is accurate when $\max_n\|{\bf a}_n\|/r\ll 1$. The phase error induced by the neglected higher-order terms is of order ${\cal O}(\chi\max_n\|{\bf a}_n\|^3/r^2)$ and is negligible for coherent processing when $\chi\max_n\|{\bf a}_n\|^3/r^2\ll 1$~\cite{lu2024tutorial}.
It is included only to illustrate the distance-dependent phase curvature; the exact Euclidean distance is used in the subsequent signal models and estimator design.
By substituting \eqref{eq:dn} into \eqref{eq:b}, the $n$-th entry of the array response vector can be approximated as
\vspace{-0.5em}\begin{equation}\label{eq:b1}
[{\bf b}(\varphi,\theta,r)]_n \!\approx\! e^{\jmath\chi\left({\bf u}^\top(\varphi,\theta){\bf a}_n-\frac{\|{\bf a}_n\|^2-[{\bf u}^\top(\varphi,\theta){\bf a}_n]^2}{2r}\right)} \!\in {\mathbb C}
\vspace{-0.25em}\end{equation}
where the terms proportional to $1/r$ characterize the distance-dependent phase curvature of the spherical wave and can be ignored when $r>r_{\rm FA}$.
This implies that the uplink incident wave exhibits a plane wavefront, and its array response vector only depends on the azimuth and elevation angles.

Throughout this paper, we adopt unit-modulus array responses to characterize the element-dependent phase and geometric variations of spherical-wave propagation.
In the considered near-field region, the spherical-spreading amplitude \(1/d_n\) varies more slowly across the array than the phase term \(e^{-\jmath\chi d_n}\), and is therefore approximated by a pathwise common amplitude~\cite{lu2024tutorial}.
Its common contribution is represented by the channel gain and the array-averaged large-scale fading (LSF) coefficients in Section~\ref{subsec:channel}, while the residual element-wise amplitude variation is neglected, since the proposed framework primarily exploits the near-field phase signatures.
}

\begin{figure}[t!]
\centering
\includegraphics[scale=1]{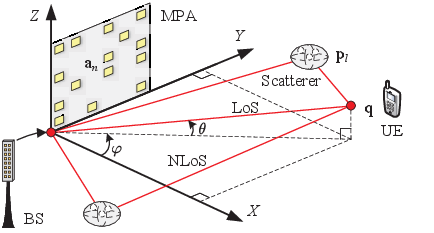}
\vspace{0em}
\caption{Considered near-field UM-MIMO system. 
\label{fig:1_system}}\vspace{-1.em}
\end{figure}

\subsection{Channel Model}\label{subsec:channel}
We adopt the block fading model where the channel coefficients can be assumed to be fixed in each coherence block.
We denote by
\vspace{-0.5em}\begin{equation}\label{eq:h}
{\bf h} = {\bf h}^{\rm L} + {\bf h}^{\rm N} \in {\mathbb C}^{N}
\vspace{-0.25em}\end{equation}
the channel transfer function between the MPA and the UE, where ${\bf h}^{\rm L}$ and ${\bf h}^{\rm N}$ are the deterministic LoS component and the stochastic NLoS component, respectively~\cite{wang2022pervasive,wang2025enhanced}.
{The LoS component is given by
\vspace{-0.5em}\begin{equation}\label{eq:hLOS}
{\bf h}^{\rm L} = \sqrt{\beta^{\rm L}}{\bf b}(\varphi_{\rm u},\theta_{\rm u},r_{\rm u})
\vspace{-0.25em}\end{equation}
where $\varphi_{\rm u}$, $\theta_{\rm u}$, and $r_{\rm u}$ are the azimuth AoA, elevation AoA, and the UE-to-origin distance, respectively. The LoS LSF coefficient is given by $\beta^{\rm L} = \frac{1}{N} ({\bf h}^{\rm L})^{\Htran}{\bf h}^{\rm L}$.}
Given that the scatterers exist in a near-field 3D region characterized by ${\cal S} = \{l  : r_l \in [r_{\min}, r_{\max}], \varphi_l \in [\varphi_{\min}, \varphi_{\max}],\theta_l \in [\theta_{\min}, \theta_{\max}], \forall l\}$, the NLoS component consists of a superposition of multipath components from the scatterers in ${\cal S}$ and is given by
\vspace{-0.5em}\begin{equation}\label{eq:hNLoS}
{\bf h}^{\rm N} \!=\! {\int_{r_{\min}}^{r_{\max}}}{\int_{\varphi_{\min}}^{\varphi_{\max}}}{\int_{\theta_{\min}}^{\theta_{\max}}} g(\varphi,\theta,r) {\bf b}(\varphi,\theta,r) {\rm d}\varphi {\rm d}\theta {\rm d}r
\vspace{-0.25em}\end{equation}
where $r_{\max} < r_{\rm FA}$ and ${\bf b}(\varphi,\theta,r)$ is the array response vector for an arbitrary scatterer.
We denote by $g(\varphi,\theta,r)$ the corresponding angular and distance spreading function determining the gain and phase-shift, which is modeled as a spatially uncorrelated circularly symmetric Gaussian stochastic process~\cite{bjornson2017massive}.
Therefore, the NLoS component ${\bf h}^{\rm N}$ can be characterized by its spatial correlation matrix
\vspace{-0.5em}\begin{align}\label{eq:RN}
&{\bf R}^{\rm N} = {\mathbb E}\{{\bf h}^{\rm N}({\bf h}^{\rm N})^{\Htran}\} \\
&\notag = \beta^{\rm N} \!\!{\int_{r_{\min}}^{r_{\max}}}\!\!{\int_{\varphi_{\min}}^{\varphi_{\max}}}\!\!{\int_{\theta_{\min}}^{\theta_{\max}}} \!\!\!f(\varphi,\theta,r) {\bf b}(\varphi,\theta,r){\bf b}^{\Htran}(\varphi,\theta,r) {\rm d}\varphi {\rm d}\theta {\rm d}r
\vspace{-0.25em}\end{align}
as ${\bf h}^{\rm N} \sim {\cal N}_{\mathbb C}({\bf 0},{\bf R}^{\rm N})$, where $\beta^{\rm N} = \frac{1}{N} {\rm tr}({\bf R}^{\rm N})$ represents the NLoS LSF coefficient and $f(\varphi,\theta,r)$ is the normalized spatial scattering function with ${\int_{r_{\min}}^{r_{\max}}}{\int_{\varphi_{\min}}^{\varphi_{\max}}}{\int_{\theta_{\min}}^{\theta_{\max}}} f(\varphi,\theta,r) {\rm d}\varphi {\rm d}\theta {\rm d}r = 1$.
Letting $\kappa$ be the Rician K-factor, we have
\vspace{-0.5em}\begin{equation}
\beta_{}^{\rm L} = {\frac{\kappa_{}}{\kappa_{}+1} }\beta_{}, \quad \beta_{}^{\rm N} = {\frac{1}{\kappa_{}+1} } \beta_{}
\vspace{-0.25em}\end{equation}
where $\beta_{}$ is the LSF coefficient describing pathloss and shadowing of the full channel $\bf h$.

\vspace{0.em}
\subsection{Signal Transmission Model}
We consider a scenario where the UE is assigned a predefined pilot sequence 
and transmits it during the uplink pilot phase of each coherence block.
The received pilot signal after despreading with the pilot sequence is given by
\cite[Sec.3]{bjornson2017massive}
\vspace{-0.5em}\begin{equation}\label{eq:received signal at AP}
{\bf y}_{}^{} = \sqrt{\rho} {\bf h} + {\bf n}\in {\mathbb C}^{N}
\vspace{-0.25em}\end{equation}
in the case of no pilot reuse, where $\rho>0$ represents the pilot SNR and ${\bf n} \sim {\cal N}_{\mathbb C}({\bf 0}, {\bf I}_N)$ is the receiver noise.
The BS estimates the LoS component ${\bf h}^{\rm L}$ and NLoS component ${\bf h}^{\rm N}$ in sequence based on the received signal ${\bf y}$, which are elaborated in Section~\ref{sec:LoS estimation} and Section \ref{sec:NLoS estimation}, respectively.
The case with pilot contamination caused by pilot reuse among multiple UEs is investigated in Section \ref{sec:NLoS estimation MU}.
This modular design allows tailored algorithms to be applied for each component, reduces computational overhead, and enables integration with environment-aware priors from channel maps.

\vspace{-0.25em}
\section{Estimation of the LoS Component}\label{sec:LoS estimation}

It is common for the NLoS component to be treated as noise when estimating the LoS component \cite{garkisch2024user}, especially at the considered millimeter-wave band where the LoS component dominates with a large Rician K-factor, e.g., $\kappa = 10$.
{Therefore, the received signal in \eqref{eq:received signal at AP} can be rewritten as
\vspace{-0.5em}\begin{equation}\label{eq:received signal at AP1}
{\bf y}_{}^{} = \sqrt{\rho} {\bf h}^{\rm L} + {\bf n}^\prime = \alpha {\pmb b}({\bf q};{\bf a}) + {\bf n}^\prime \in {\mathbb C}^{N}
\vspace{-0.25em}\end{equation}
where ${\bf n}^\prime = \sqrt{\rho}{\bf h}^{\rm N}+{\bf n}$ and $\alpha=\sqrt{\rho\beta^{\rm L}}e^{\jmath\chi r_{\rm u}}$ is the equivalent complex LoS gain including the reference phase associated with the UE-to-origin distance\footnote{{The coefficient $\alpha$ is an effective complex gain in the received LoS signal model rather than a LSF coefficient like $\beta$, $\beta^{\rm L}$, and $\beta^{\rm N}$ in Section~\ref{subsec:channel}}}, and thus, $|\alpha|^2=\rho\beta^{\rm L}$.
Here, for the UE position ${\bf q}$ and antenna position vector ${\bf a}$, we define the absolute-distance response ${\pmb b}({\bf q};{\bf a})$ by $[ {\pmb b}({\bf q};{\bf a})]_n=e^{-\jmath\chi d_{{\rm u},n}}$, where $d_{{\rm u},n}=\|{\bf q}-{\bf a}_n\|$, so that ${\pmb b}({\bf q};{\bf a})=e^{-\jmath\chi r_{\rm u}}{\bf b}(\varphi_{\rm u},\theta_{\rm u},r_{\rm u})$.

Since ${\bf h}^{\rm N}\sim{\cal N}_{\mathbb C}({\bf 0},{\bf R}^{\rm N})$, the combined disturbance ${\bf n}^{\prime}=\sqrt{\rho}{\bf h}^{\rm N}+{\bf n}$ follows ${\cal N}_{\mathbb C}({\bf 0},\rho{\bf R}^{\rm N}+{\bf I}_N)$, which is generally spatially correlated.
For tractable LoS localization, we use ${\bf n}^{\prime}\approx{\bf n}_{\rm w}\sim{\cal N}_{\mathbb C}({\bf 0},\sigma^2{\bf I}_N)$ with the same average power per antenna~\cite{yuan2024scalable}, where
\vspace{-0.5em}\begin{equation}\label{eq:effective_variance}
	\sigma^2 = \frac{1}{N} \text{tr}(\rho {\mathbf{R}^{\rm N}} + {\mathbf{I}_N}) = \frac{\rho}{N} \text{tr}({\mathbf{R}^{\rm N}}) + 1
\vspace{-0.25em}\end{equation}
This approximation preserves the conditional mean of the received signal and the average disturbance power while ignoring the NLoS spatial correlation.
It is exact when ${\bf R}^{\rm N}$ is proportional to ${\bf I}_N$.
More generally, it is expected to work well when the LoS component dominates, as characterized by a large Rician K-factor, or when \(\rho{\bf R}^{\rm N}\) is small relative to \({\bf I}_N\), such that the combined disturbance covariance is close to a scaled identity matrix.

We denote by ${\bf a} = [ {\bf a}_1^\top, \ldots ,{\bf a}_N^\top ]^\top$ the collective position vector of the movable antennas.
Conditioned on the UE position ${\bf q}$, the movable antenna positions ${\bf a}$, and the equivalent complex channel gain $\alpha$, the observation $\mathbf{y}$ follows $\mathcal{N}_{\mathbb C}\left( {\alpha {\pmb b}({\bf q};{\bf a}), {\sigma ^2}\mathbf{I}_N} \right)$.
Therefore, the maximum likelihood (ML) estimates of the UE position ${\bf q}$ and the complex channel gain $\alpha$ are given by~\cite{kay1993fundamentals}
\vspace{-0.5em}\begin{equation}
({\hat {\bf q}}, {\hat \alpha}) = \mathop {\arg \max }\limits_{{\bf q},\alpha }  - \frac{1}{{{\sigma ^2}}}{\left\| {\mathbf{y} - \alpha {\pmb b}({\bf q};{\bf a})} \right\|^2}.
	\label{e5}
\vspace{-0.25em}\end{equation}

To characterize the LoS localization geometry, we derive a local FILB under \eqref{eq:effective_variance} by freezing $\sigma^2$ at its value at the operating point.
	
\subsection{Fisher Information-Based Lower-Bound Benchmark}

Consider the signal model in \eqref{eq:received signal at AP1}.
We first recall from \cite{collier2005fisher} the standard Fisher information matrix (FIM) for a complex Gaussian observation with a real-valued parameter vector and then specialize it to UE position estimation.

\begin{lemm}\label{lemm:FIM}
For an $N$-dimensional complex Gaussian observation signal ${\bf y} \sim {\cal N}_{\mathbb C}({\boldsymbol{\mu}}({\bf q}),{\boldsymbol{\Sigma}}({\bf q}))$, the $\left( {i,j} \right)$-th entry of the FIM ${\bf J}\in {\mathbb R}^{3\times 3}$ for estimating the real UE position vector $\mathbf{q}$ can be computed via
\vspace{-0.5em}\begin{equation}
	\begin{aligned}
	{[ {{\mathbf{J}}} ]_{ij}}
	={}&2\Re \left\{
	\frac{{\partial {{\left( {\boldsymbol{\mu}} \right)}^\text{H}}}}{{\partial [{\bf q}]_i}}
	{{\left( {{\boldsymbol{\Sigma}}} \right)}^{ - 1}}
	\frac{{\partial \boldsymbol{\mu}}}{{\partial [{\bf q}]_j}}
	\right\} \\
	&+\mathrm{tr}\left(
	{{\left( {{\boldsymbol{\Sigma}}} \right)}^{ - 1}}
	\frac{{\partial {\boldsymbol{\Sigma}}}}{{\partial [{\bf q}]_i}}
	{{\left( {{\boldsymbol{\Sigma}}} \right)}^{ - 1}}
	\frac{{\partial {\boldsymbol{\Sigma}}}}{{\partial [{\bf q}]_j}}
	\right).
	\end{aligned}
		\label{e6}
\vspace{-0.25em}\end{equation}
\end{lemm}

Applying Lemma \ref{lemm:FIM}, we set ${\boldsymbol{\mu}}=\alpha{\pmb b}({\bf q};{\bf a})$ and $\boldsymbol{\Sigma}=\sigma^2\mathbf{I}_N$ for ${\bf q}=[x_{\rm u},y_{\rm u},z_{\rm u}]^\top$.
Although ${\bf R}^{\rm N}$ and hence $\sigma^2$ may depend on ${\bf q}$ through pathloss and local scattering, the local FILB freezes $\sigma^2$ at the operating point, which is appropriate when the average NLoS power varies slowly within the localization region.
Thus, $\frac{{\partial \boldsymbol{\Sigma}}}{{\partial {[{\bf q}]_i}}}=\mathbf{0}$ and the FIM in \eqref{e6} simplifies to
\vspace{-0.5em}\begin{equation}
{\left[ {{\mathbf{J}}} \right]_{ij}} = \frac{2}{\sigma^2}\Re \left\{ {\frac{{\partial {{\left( {\boldsymbol{\mu}} \right)}^\text{H}}}}{{\partial {[{\bf q}]_i}}} \frac{{\partial \boldsymbol{\mu}}}{{\partial {[{\bf q}]_j}}}} \right\}.
	\label{e7}
\vspace{-0.25em}\end{equation}
Therefore, the derivative of ${\boldsymbol{\mu}}$ with respect to the $i$-th coordinate of the UE position ${\bf q}$ is given by
\vspace{-0.5em}\begin{align}	\label{e8}
\frac{{\partial \boldsymbol{\mu}}}{{\partial {[{\bf q}]_i}}} &= -\jmath\chi \alpha \\
\notag &\cdot \mathrm{diag}\left(\frac{[{\bf q}]_i - [{\bf a}_{1}]_i}{d_{{\rm u},1}}, \dots, \frac{[{\bf q}]_i - [{\bf a}_{N}]_i}{d_{{\rm u},N}}\right) {\pmb b}({\bf q};{\bf a}).
\vspace{-0.25em}\end{align}
Substituting \eqref{e8} into the FIM expression in \eqref{e7}, we have
\vspace{-0.5em}\begin{equation}
	{\left[ \mathbf{J} \right]_{ij}} = \frac{{2{{\left| \alpha \right|}^2} \chi^2 }}{{{\sigma ^2}}}\sum\limits_{n = 1}^N {\frac{{\left( {[{\bf q}]_i - [{\bf a}_{n}]_i} \right)\left( {[{\bf q}]_j - [{\bf a}_{n}]_j} \right)}}{d_{{\rm u},n}^2}}
	\label{e9}
\vspace{-0.25em}\end{equation}
and the FIM in a compact matrix form as
\vspace{-0.5em}\begin{equation}
	\mathbf{J} = \frac{{2{{\left| \alpha  \right|}^2} \chi^2 }}{{{\sigma ^2}}}\sum\limits_{n = 1}^N {\frac{{\left( {{\bf q} - {\bf a}_n} \right){{\left( {{\bf q} - {\bf a}_n} \right)}^\top}}}{d_{{\rm u},n}^2}}.
	\label{e10}
\vspace{-0.25em}\end{equation}
For the adopted local benchmark with fixed $\alpha$ and $\sigma^2$, the covariance of any unbiased position estimate $\hat{\bf q}$ is lower bounded in the positive semidefinite sense by~\cite{collier2005fisher}
\vspace{-0.5em}\begin{equation}
\mathbb{E}\{ {( {\hat{\bf q}- {\bf q}} ){{( {\hat{\bf q} - {\bf q}} )}^\top}} \} \succeq {\mathbf{J}^{ - 1}}.
\vspace{-0.25em}\end{equation}

Under the adopted local model with fixed $\sigma^2$, we define the benchmark for the position RMSE as
${\sf FILB}\triangleq\sqrt{\mathrm{tr}({\bf J}^{-1})}$.
If the effective covariance varies appreciably with ${\bf q}$, the covariance derivative term in \eqref{e6} must be retained.
Accordingly, the FILB based on \eqref{e10} is a local benchmark for antenna geometry design, rather than the Bayesian Cram{\'e}r--Rao lower bound for the full model with position-dependent covariance.

As shown in \eqref{e10}, ${\bf J}$ is a sum of rank-one matrices determined by the antenna-to-UE directions.
Diverse observation directions improve the conditioning of ${\bf J}$, whereas clustered directions make it nearly singular.
We therefore adopt the D-optimality criterion, which maximizes $\det({\bf J})$~\cite{zamir1998proof}.
Since the fixed $\sigma^2$ only scales ${\bf J}$, it does not affect the antenna placement that maximizes this criterion.
The resulting optimization problem is formulated as}
\begin{subequations}
	\vspace{-0.5em}\begin{align}
		{\sf P}_{1}: \quad
		& \mathop {\max }\limits_{\{ {{\bf a}_n} \in {\mathbb R}^3\}} \log {\det(\mathbf{J} )} \label{e11a}\\
		{\rm s.t.} \quad & {\bf a}_n \in \mathcal{Q},\ n = 1,\ldots,N, \label{e11b} \\
		&  \|{\bf a}_m-{\bf a}_n\| \ge d,\ m,n = 1,\ldots,N,\ m\neq n, \label{e11c}
	\vspace{-0.25em}\end{align}
\end{subequations}
where~\eqref{e11b} confines each movable antenna to a feasible region $\mathcal{Q}$ and~\eqref{e11c} enforces a minimum separation $d$ to avoid collisions.
{This antenna placement optimization is designed to improve the FIM conditioning of subsequent near-field LoS estimation.}
Problem ${\sf P}_{1}$ is a challenging non-convex optimization problem.
{To this end, we propose an iterative algorithm based on the projected gradient ascent (PGA) algorithm~\cite{xu2024efficient}, augmented with a penalty method to handle the non-convex inter-antenna distance constraint.}

\vspace{-0.5em}
\subsection{Channel Map-based PGA Localization}
{In practice, the exact UE position ${\bf q}$ is not available at the BS.
To this end, we leverage a channel map, namely UDM, to provide an initial and coarse estimate ${\tilde{\bf q}}$ based on the environment database $\cal E$, which is expressed as
\vspace{-0.5em}\begin{equation}
	{\cal M}_{\rm UDM}:\quad
	{\tilde{\bf q}} = f_{\rm UDM}({\cal E})
	\vspace{-0.25em}\end{equation}
and constructed via approaches like channel measurements.
The provided coarse UE position is given by
\vspace{-0.5em}\begin{equation}
	{\tilde{\bf q}} = {\bf q} + {\boldsymbol \delta}_{\rm UDM}
	\vspace{-0.25em}\end{equation}
where ${\boldsymbol \delta}_{\rm UDM} = \epsilon_{\rm UDM}({\boldsymbol\zeta}_{\rm UDM}\odot{\bf q})$ is the deviation, $\epsilon_{\rm UDM}$ is the error of the ${\cal M}_{\rm UDM}$, and ${\boldsymbol\zeta}_{\rm UDM}\in {\mathbb R}^3$ is a random vector with each entry $[{\boldsymbol\zeta}_{\rm UDM}]_i \sim {\cal U} (-{1}/{2}, {1}/{2})$, $i=1,\ldots,3$.
In the following PGA design, evaluating the FIM in \eqref{e10} at the UDM-provided \(\widetilde{\bf q}\) naturally couples the UDM prior with the PGA antenna placement.
The PGA optimization and mechanical repositioning processes are performed on a slower configuration timescale rather than within each coherence block.
Once repositioning is complete, the optimized geometry is reused over multiple coherence blocks, avoiding repeated mechanical movement in each block. Fresh pilots are acquired in every block rather than relying on the channel observed before movement.
This operating mode is intended for quasi-static or low-mobility UEs whose coarse positions remain valid during reconfiguration and geometry reuse, rather than for high-mobility scenarios requiring per-block adaptation.}

We next reformulate problem ${\sf P}_{1}$.
More precisely, to handle the non-convex constraint~\eqref{e11c}, we employ a penalty method, which transforms the constrained problem into a more tractable form by incorporating the non-convex constraint into the objective function.
This approach approximates problem ${\sf P}_{1}$ as
\begin{subequations}
	\vspace{-0.5em}\begin{align}
		\label{e12a}{\sf P}_{2}: \quad
		&\mathop {\max }\limits_{\{ {{\bf a}_n} \in {\mathbb R}^3\}} \log {\det(\mathbf{J})} \\
		& \notag \qquad \qquad - \frac{\gamma }{2} \sum_{1\leq m<n\leq N} \left[ \max(0, d - {\left\| {\mathbf{a}_m - \mathbf{a}_n} \right\|}) \right]^2 \\
		{\rm s.t.} \quad & \eqref{e11b}
	\vspace{-0.25em}\end{align}
\end{subequations}
where $\gamma > 0$ is a constant penalty coefficient, and each unordered antenna pair is counted once.
The quadratic penalty term is zero when the constraint is satisfied and grows rapidly as the violation increases.

Problem ${\sf P}_{2}$ can now be effectively solved via the PGA algorithm.
The PGA algorithm iteratively updates the antenna positions by taking a step in the direction of the gradient of the objective function~\eqref{e12a}, followed by a projection onto the feasible set $\mathcal{Q}$.
The crucial step is to derive the gradient of the objective function with respect to the position of each antenna $\mathbf{a}_n$.
Specifically, the total gradient can be decomposed into two components, one from the D-optimality objective, denoted as $\boldsymbol{\varrho}^{\mathrm{obj}}_n$, and one from the penalty term, denoted as $\boldsymbol{\varrho}^{\mathrm{pen}}_n$.
The gradient of the primary objective $\log\det(\mathbf{J})$ follows from Jacobi's formula, as stated in the following lemma.

\vspace{-0.5em}
\begin{lemm}\label{lemm:D_optimality}
	The gradient of the D-optimality objective function with respect to the position of the $n$-th antenna, $\mathbf{a}_n$, is given by $\boldsymbol{\varrho}^{\mathrm{obj}}_n = \nabla_{\mathbf{a}_n} \log\det(\mathbf{J})$, where its $j$-th coordinate is
	\vspace{-0.5em}\begin{equation}
		[\boldsymbol{\varrho}^{\mathrm{obj}}_n]_j = \frac{\partial (\log\det(\mathbf{J}))}{\partial [\mathbf{a}_n]_j } =  \mathrm{tr}\left( \mathbf{J}^{-1} \frac{\partial \mathbf{J}}{\partial [\mathbf{a}_{n}]_j} \right).
	\vspace{-0.25em}\end{equation}
	The derivative of the FIM, $\frac{\partial \mathbf{J}}{\partial [\mathbf{a}_{n}]_j}$, is given by~\eqref{e15} at the bottom of the next page, where ${\hat{\bf e}}_j$ is the standard basis vector for the $j$-th coordinate.
\end{lemm}
\begin{figure*}
	\vspace{-0.5em}\begin{equation}
		{\frac{\partial \mathbf{J}}{\partial  [\mathbf{a}_{n}]_j} = \frac{\xi}{ \left\| {{\tilde{\bf q}} - {\mathbf{a}_n}} \right\|^4 } \left[ 2([{\tilde{\bf q}}]_j -  [\mathbf{a}_{n}]_j)({\tilde{\bf q}}-\mathbf{a}_n)({\tilde{\bf q}}-\mathbf{a}_n)^\top - \left\| {{\tilde{\bf q}} - {\mathbf{a}_n}} \right\|^2 ({\hat{\bf e}}_j({\tilde{\bf q}}-\mathbf{a}_n)^\top + ({\tilde{\bf q}}-\mathbf{a}_n){\hat{\bf e}}_j^\top) \right]}
		\label{e15}
	\vspace{-2.em}\end{equation}
\end{figure*}
\renewcommand\qedsymbol{$\blacksquare$}
\begin{proof}
	The details are relegated to Appendix \ref{appe:D_optimality}.
\end{proof}
For closely spaced antennas, the penalty gradient with respect to $\mathbf{a}_n$ is
\vspace{-0.5em}\begin{equation}
	\boldsymbol{\varrho}^{\mathrm{pen}}_n = \gamma \sum_{\substack{m \neq n\\0<\left\|\mathbf{a}_m-\mathbf{a}_n\right\|<d}} \left( d - {\left\| {\mathbf{a}_m - \mathbf{a}_n} \right\|} \right) \frac{\mathbf{a}_m - \mathbf{a}_n}{\left\| {\mathbf{a}_m - \mathbf{a}_n} \right\|}.
\vspace{-0.25em}\end{equation}
Subtracting this term from the primary-objective gradient produces the repulsive update. Thus, the total gradient of the objective function in $\mathrm{P}_{2}$ with respect to $\mathbf{a}_n$ is expressed as
\vspace{-0.5em}\begin{equation}
	\boldsymbol{\varrho}_n = \boldsymbol{\varrho}^{\mathrm{obj}}_n - \boldsymbol{\varrho}^{\mathrm{pen}}_n.
	\label{e16}
\vspace{-0.25em}\end{equation}
Letting $\boldsymbol{\varrho} = {\left[ {\boldsymbol{\varrho}}_1^\top, \ldots ,{\boldsymbol{\varrho}}_N^\top \right]^\top}$, we update the antenna position vector $\mathbf{a}$ iteratively as
\vspace{-0.5em}\begin{equation}
	{\mathbf{a}^{(i)}} = {\mathbf{a}^{(i - 1)}} + \vartheta  \cdot {\boldsymbol{\varrho}^{(i-1)}}
	\label{e17}
\vspace{-0.25em}\end{equation}
where $\vartheta$ is a fine-tuned step size.
The iteration stops when $\left| {\log \det \left( {{\mathbf{J}^{(i)}}} \right) - \log \det \left( {{\mathbf{J}^{(i - 1)}}} \right)} \right|$ in one iteration is less than the convergence threshold ${\varepsilon}$, or when the maximum number of iterations ${\mathrm{I}}_{\max}$ is reached.
The overall channel map-based PGA algorithm for antenna position optimization is summarized in Algorithm 1.

\begin{algorithm}[t!]
	\caption{Channel Map-based PGA for $\mathrm{P}_{2}$}
	\label{alg3}
	\textbf{Input:} $\mathbf{a}^{(0)}$, $\vartheta $, $\gamma$, ${\tilde{\bf q}}$, ${\mathrm{I}}_{\max}$, ${\varepsilon}$, $d$, $\mathcal{Q}$.\\
	\textbf{Output:} Optimized antenna positions $\mathbf{a}^{\star}$.\\
	Compute the objective $\log {\det(\mathbf{J}^{(0)} )}$\;
	\For{$i = 1$ to ${\mathrm{I}}_{\max}$}{
		\For{$n = 1$ to $N$}{
			Compute $\boldsymbol{\varrho}^{(i-1)}_n$ at $\mathbf{a}^{(i-1)}$ according to~\eqref{e16}\;
		}
		Update the antenna position $\mathbf{a}^{(i)}$ according to~\eqref{e17}\;
		Project the antenna position $\mathbf{a}^{(i)}$ to $\mathcal{Q}$\;
		Compute the objective $\log {\det(\mathbf{J}^{(i)} )}$\;
		\If{$ \left| {\log \det \left( {{\mathbf{J}^{(i)}}} \right) - \log \det \left( {{\mathbf{J}^{(i - 1)}}} \right)} \right| \leqslant \varepsilon $}{\textbf{break};}
	}
	Set $\mathbf{a}^{\star}=\mathbf{a}^{(i)}$.
\end{algorithm}

{With the optimized antenna positions $\mathbf{a}^{\star}$, we now rewrite the ML estimate of the UE position and the complex channel gain $\alpha$ as
\vspace{-0.5em}\begin{equation}
	( {\hat{\bf q},\hat \alpha } ) = \mathop {\arg \max }\limits_{{\bf q},\alpha }  - \frac{1}{{{\sigma ^2}}}{\left\| {\mathbf{y} - \alpha {\pmb b}({\bf q};{\bf a}^\star) } \right\|^2}.
\vspace{-0.25em}\end{equation}
Then the optimal ${\alpha}$ that maximizes this objective for a given candidate ${\bf q}$ has a closed-form solution ${\hat \alpha}({\bf q}) =  {\pmb b}^{\Htran}({\bf q};{\bf a}^\star) {\mathbf{y}}/{\left\| {\pmb b}({\bf q};{\bf a}^\star) \right\|^2}$.
By substituting this result back into \eqref{e5}, the joint optimization problem reduces to a search over ${\bf q}$ only. This is equivalent to maximizing the profiled log-likelihood function, which we denote as $\mathcal{L}\left( { {\bf q},{\mathbf{a}^ \star}} \right)$}
\vspace{-0.5em}\begin{equation}\label{eq:likelihood}
{\mathop {\max }\limits_{{\bf q}} \mathcal{L}\left( { {\bf q},{\mathbf{a}^ \star}} \right) = - \frac{1}{{{\sigma ^2}}}{\left\| {\mathbf{y} - \frac{{{\pmb b}^{\Htran}({\bf q};{\bf a}^\star)\mathbf{y}}}{{{{\left\| {\pmb b}({\bf q};{\bf a}^\star) \right\|}^2}}}{\pmb b}({\bf q};{\bf a}^\star)} \right\|^2}.}
\vspace{-0.25em}\end{equation}

The optimized antenna placement improves the likelihood geometry, enabling $\hat{\bf q}$ to be obtained by grid search.
We define a uniform $D$-dimensional grid $\mathcal{G}$ that spans the search region $[ - \eta ,\eta ]^D$ with $G$ points along each dimension.
The estimated UE position is then obtained by exhaustively searching for the point within this grid that maximizes the profiled log-likelihood function
\vspace{-0.5em}\begin{equation}
	\hat{\bf q} = \mathop {\arg \max }\limits_{{\bf q} \in \mathcal{G}} \mathcal{L}\left( { {\bf q},{\mathbf{a}^ \star}} \right).
\vspace{-0.25em}\end{equation}
By substituting this result back into the LoS model, we can obtain the LoS estimate ${\widehat{\bf h}}^{\rm L}=(\hat{\alpha}/\sqrt{\rho}){\pmb b}(\hat{\bf q};{\bf a}^\star)$.
{The PGA algorithm optimizes antenna positions rather than the UE position.
A UDM error or a suboptimal placement may reduce geometric conditioning, but the UE position is estimated afterward from fresh pilots by the LoS estimator.}

Finally, the localization performance is evaluated using the root MSE (RMSE) of the estimates of ${\bf q}$, defined as
\vspace{-0.25em}\begin{equation}
	{\sf RMSE} = \sqrt {\mathbb{E}\{ {\| {{\bf q} - {{\hat{\bf q}}}} \|^2} \}} .
\vspace{-0.25em}\end{equation}

\section{Estimation of the NLoS Component}\label{sec:NLoS estimation}

{After the estimated LoS component is subtracted from the received signal, \({\bf h}^{\rm N}\) becomes the channel component to be estimated, and its spatial structure is therefore explicitly exploited in this stage.
The residual observation can be written explicitly as
${\bf y}^{\rm N}=\sqrt{\rho}{\bf h}^{\rm N}+\sqrt{\rho}({\bf h}^{\rm L}-{\widehat{\bf h}}^{\rm L})+{\bf n}\in {\mathbb C}^{N}$,
where the second term is the residual LoS component caused by localization and channel reconstruction errors.}

Various channel estimators can be used under different levels of prior information and different accuracy-complexity requirements. Their accuracy is quantified by the normalized MSE (NMSE) ${\mathbb E}\{\|{\bf h}^{\rm N} - {\widehat{\bf h}}^{\rm N}\|^2\}/{\mathbb E}\{\|{\bf h}^{\rm N}\|^2\}$.
When the spatial correlation matrix ${\bf R}^{\rm N}$ is available at the BS, the Bayesian \emph{linear minimum MSE (MMSE)} estimator \cite{bjornson2017massive} minimizes ${\mathbb E}\{\|{\bf h}^{\rm N} - {\widehat{\bf h}}^{\rm N}\|^2\}$ and is given by~\cite[Sec. 3]{bjornson2017massive}
\vspace{-0.5em}\begin{equation}\label{eq:mmse estimate}
  {\widehat{\bf h}}^{\rm N}_{\rm MMSE} = \sqrt{\rho} {\bf R}^{\rm N} ({\rho} {\bf R}^{\rm N} + {\bf I}_N)^{-1} {\bf y}^{\rm N}. 
\vspace{-0.25em}\end{equation}
If the BS has no prior information regarding the spatial correlation matrix ${\bf R}^{\rm N}$, the non-Bayesian LS estimator can be employed to minimize $\|{\bf y}^{\rm N} - {\sqrt \rho}{\widehat{\bf h}}^{\rm N}\|^2$, which is given by ${\widehat{\bf h}}^{\rm N}_{\rm LS} = {\bf y}^{\rm N}/{\sqrt{\rho}}$.
The MMSE estimator is challenging to implement in practice because acquiring the ${N(N+1)}/{2}$ distinct entries of ${\bf R}^{\rm N}$ and inverting $({\rho} {\bf R}^{\rm N} + {\bf I}_N) \in {\mathbb C}^{N \times N}$ incur prohibitive overhead and complexity for UM-MIMO systems with large $N$.
The LS estimator is applicable for practical implementation by only utilizing the pilot SNR. However, it is sensitive to noise and overly conservative since it fails to exploit readily available information such as the MPA's antenna deployment.

The RSLS estimator \cite{demir2022channel} achieves a good trade-off between the implementation feasibility and estimation accuracy by exploiting the array geometry of the considered planar array.
More precisely, the received pilot signal ${\bf y}^{\rm N}$ is first projected onto a reduced-dimensional eigenspace of ${\bf R}^{\rm N}$ as ${{\bf U}}^{\Htran}{\bf y}^{\rm N}$.
The semi-unitary matrix ${\bf U} \in {\mathbb C}^{N\times r}$ contains the orthonormal eigenvectors in the compact EVD ${\bf R}^{\rm N} = {\bf U}{\bf \Lambda}{\bf U}^{\Htran}$, where $r={\rm Rank}({\bf R}^{\rm N})\ll N$ and the diagonal matrix ${\bf \Lambda}\in {\mathbb C}^{r\times r}$ contains the non-negligible eigenvalues.
Then the channel estimate is obtained via the LS approach and projected back into the original domain, which is given by
\vspace{-0.5em}\begin{equation}\label{eq:rsls estimate}
  {\widehat{\bf h}}^{\rm N}_{\rm RSLS} = \frac{{{\bf U}}{{\bf U}}^{\Htran}}{\sqrt{\rho}}{\bf y}^{\rm N}.
\vspace{-0.25em}\end{equation}
{The RSLS estimator in \eqref{eq:rsls estimate} projects \({\bf y}^{\rm N}\) onto the adopted NLoS subspace and approaches the MMSE estimator as \(\rho\rightarrow\infty\)~\cite{demir2022channel}.
Under imperfect LoS cancellation, the same projection attenuates the out-of-subspace residual LoS component, so its downstream impact is determined by the residual energy retained within the subspace.}
Since the real ${\bf R}^{\rm N}$ is hard to obtain, one can alternatively utilize a tractable ${\bar{\bf R}}^{\rm N}$ with only a one-dimensional integral for RSLS estimation~\cite{demir2024spatial}.
The eigenspace of ${\bar{\bf R}}^{\rm N}$ contains that of ${\bf R}^{\rm N}$.
In this case, matrix $\bf U$ in~\eqref{eq:rsls estimate} is replaced by ${\bar{\bf U}}\!\in\! {\mathbb C}^{N\times {\bar r}}$, which is the eigenvector matrix 
of ${\bar{\bf R}}^{\rm N}$ with ${\bar r} = {\rm Rank}({\bar{\bf R}}^{\rm N})$.
Nevertheless, this approach turns out to be conservative if the positions of scatterers $\{{\bf p}_l:\forall l\}$ are available.

\vspace{0.em}
\subsection{Sketch Algorithm-Based RSLS Channel Estimation}\label{subsec:sa rsls}

A critical step in the aforementioned RSLS estimator is the acquisition of the signal subspace matrix $\bf U$ (or ${\bar {\bf U}}$), which relies on a compact EVD of the $N \times N$ spatial correlation matrix ${\bf R}^{\rm N}$ (or ${\bar{\bf R}}^{\rm N}$), with computational complexity of ${\cal O}(N^3)$.
This cubic complexity becomes prohibitive in UM-MIMO scenarios (e.g., $N = 1024$), exceeding the real-time capability of practical hardware.

To address this bottleneck, we propose a SA-RSLS channel estimator, which embeds the representative sketching algorithm \cite{martinsson2020randomized} into the RSLS framework.
The core idea of the SA-RSLS estimator is to project the high-dimensional low-rank ${\bf R}^{\rm N}$ onto a low-dimensional subspace via random sketching, thereby converting the full-dimensional EVD into a low-complexity EVD of a small projection matrix.
As a result, its complexity is reduced from ${\cal O}(N^3)$ to ${\cal O}(NN_{\rm s}{\tilde r})$, where $N_{\rm s} = {{\tilde r}+s}$ with ${\tilde r}$ being the sketch size and $s \in [5,10]$ being the oversampling factor used to bound the approximation error.
Both ${\tilde r}$ and $N_{\rm s}$ are much smaller than $N$, making it suitable for UM-MIMO.

The proposed SA-RSLS estimator operates via two core phases, i.e., sketch-based subspace acquisition and reduced-space channel estimation. Detailed steps are as follows:
\begin{enumerate}
  \item Random Sketch Generation. Construct matrix ${\bf \Omega} \in {\mathbb C}^{N \times N_{\rm s}}$ randomly with each entry $[{\bf \Omega}]_{ij} \sim {\cal N}_{\mathbb C}(0,1)$, $\forall i,j$.
      Generate sketch ${\bf S} = {\bf R}^{\rm N}{\bf \Omega} \in {\mathbb C}^{N \times N_{\rm s}}$, whose column space approximates that of ${\bf R}^{\rm N}$.
  \item Low-Dimensional Basis Extraction. Perform compact QR decomposition (QRD) on ${\bf S}$ as ${\bf S} = {{\bf Q}_{\rm s}}{{\bf\Delta}_{\rm s}}$, where ${{\bf Q}_{\rm s}} \in {\mathbb C}^{N \times N_{\rm s}}$ with ${{\bf Q}_{\rm s}}^{\Htran} {{\bf Q}}_{\rm s} = {\bf I}_{N_{\rm s}}$ and ${{\bf\Delta}_{\rm s}} \in {\mathbb C}^{N_{\rm s} \times N_{\rm s}}$ is an upper triangular matrix.
      Matrix ${{\bf Q}_{\rm s}}$ approximates the signal subspace of ${\bf R}^{\rm N}$ with minimal overhead.
  \item Low-Dimensional EVD. Project ${\bf R}^{\rm N}$ onto the subspace of ${{\bf Q}_{\rm s}}$ to get ${\bf X} = {{\bf Q}}_{\rm s}^{\Htran}{\bf R}^{\rm N}{{\bf Q}}_{\rm s} \in {\mathbb C}^{N_{\rm s} \times N_{\rm s}}$. Perform compact EVD on ${\bf X}$ as ${\bf X} = {\bf U}_{\rm x} {\bf \Lambda}_{\rm x} {\bf U}_{\rm x}^{\Htran}$, where ${\bf U}_{\rm x} \in {\mathbb C}^{N_{\rm s}\times {\tilde r}}$ and ${\bf \Lambda}_{\rm x} \in {\mathbb C}^{{\tilde r} \times {\tilde r}}$.
      This compact EVD is far less complex than decomposing ${\bf R}^{\rm N}$ directly.
  \item Subspace Recovery. Recover ${\bf U}_{\rm s} = {\bf Q}_{\rm s} {\bf U}_{\rm x}\in {\mathbb C}^{N \times {\tilde r}}$ to approximate the signal subspace of ${\bf R}^{\rm N}$.
  \item Reduced-Space Estimation: Project ${\bf y}^{\rm N}$ onto ${\bf U}_{\rm s}$ as ${{\bf U}}_{\rm s}^{\Htran}{\bf y}^{\rm N}$.
      Apply LS estimation and project back into the original domain, which is given by
        \vspace{-0.5em}\begin{equation}\label{eq:sa_rsls estimate}
          {\widehat{\bf h}}^{\rm N}_{\rm SA-RSLS} = \frac{{{\bf U}_{\rm s}}{{\bf U}}_{\rm s}^{\Htran}}{\sqrt{\rho}}{\bf y}^{\rm N}.
        \vspace{-0.25em}\end{equation}
\end{enumerate}
{The interpretation of residual LoS following \eqref{eq:rsls estimate} also applies to SA-RSLS, with ${\bf U}$ replaced by ${\bf U}_{\rm s}$.

Since the sketch is formed from ${\bf R}^{\rm N}$ rather than ${\bf y}^{\rm N}$, low SNR does not directly perturb ${\bf U}_{\rm s}$ when ${\bf R}^{\rm N}$ is available, but increases the noise retained after projection.
The accuracy and complexity are mainly controlled by ${\tilde r}$ and $s$.
More precisely, choosing ${\tilde r}$ below the effective channel rank causes projection loss, an unnecessarily large ${\tilde r}$ retains additional noise, and increasing $s$ improves subspace capture under weak spectral separation at additional cost~\cite{martinsson2020randomized}.
With $N_{\rm s}={\tilde r}+s$, the dominant matrix multiplications scale as ${\cal O}(NN_{\rm s}{\tilde r})$.
For $N=1024$, ${\tilde r}=10$, and $N_{\rm s}=18$, this corresponds to approximately ${\cal O}(1.8\times10^5)$ operations, three orders of magnitude lower than the ${\cal O}(1.1\times10^9)$ complexity of full-dimensional RSLS.}
\vspace{-1em}
\subsection{Channel Map-Based RSLS Channel Estimation}

To further improve the estimation accuracy and reduce the overhead and computational complexity, we propose a CM-RSLS channel estimator.
A channel map, namely SDM, is proposed to provide the prior positions of scatterers $\{{\hat{\bf p}}_l:\forall l\}$ based on the environment database $\cal E$, which is expressed as
\vspace{-0.5em}\begin{equation}
{\cal M}_{\rm SDM}:\quad
\{{\hat{\bf p}}_l:\forall l\} = f_{\rm SDM}({\cal E})
\vspace{-0.25em}\end{equation}
and constructed via approaches like channel measurements.
{Due to the imperfection of the ${\cal M}_{\rm SDM}$, the provided location of scatterer $l$ is given by
\vspace{-0.5em}\begin{equation}
{\hat{\bf p}}_l = {\bf p}_l + {\boldsymbol \delta}_{{\rm SDM},l}
\vspace{-0.25em}\end{equation}
where ${\boldsymbol \delta}_{{\rm SDM},l} = \epsilon_{\rm SDM}({\boldsymbol\zeta}_{{\rm SDM},l}\odot{\bf p}_l)$ is the deviation, the entries of ${\boldsymbol\zeta}_{{\rm SDM},l}$ are independently distributed as ${\cal U}(-1/2,1/2)$, and $\epsilon_{\rm SDM}$ is the error level of ${\cal M}_{\rm SDM}$.
Using the true scatterer locations, the physical NLoS component ${\bf h}^{\rm N}$ in \eqref{eq:hNLoS} can be expressed as a discrete superposition of multipaths,
\vspace{-0.5em}\begin{equation}\label{eq:h1}
  {\bf h}^{\rm N} = \sum\nolimits_{l=1}^L g_l {\bf b}({\varphi}_l,{\theta}_l,{r}_l)
\vspace{-0.25em}\end{equation}
where $({\varphi}_l,{\theta}_l,{r}_l)$ corresponds to ${\bf p}_l$ and $g_l \sim {\cal N}_{\mathbb C}(0, \beta_l^{\rm N})$ indicates the gain and phase shift of the multipath from scatterer $l$, with $\sum_{l=1}^L \beta_l^{\rm N} = \beta^{\rm N}$.

Using the SDM-provided locations $({\hat\varphi}_l,{\hat\theta}_l,{\hat r}_l)$, we construct a map-derived subspace intended to approximate the eigenspace of the actual spatial correlation matrix ${\bf R}^{\rm N}$.}
\begin{thm}\label{theo:decomposition}
{Let ${\bf U}^{\prime}=[{\bf b}(\varphi_1,\theta_1,r_1),\ldots,{\bf b}(\varphi_L,\theta_L,r_L)]$ denote the response matrix of the true scatterer locations and let ${\hat{\bf U}}=[{\bf b}({\hat\varphi}_1,{\hat\theta}_1,{\hat r}_1),\ldots,{\bf b}({\hat\varphi}_L,{\hat\theta}_L,{\hat r}_L)]={\bf Q}{\bf \Delta}$ denote the compact QRD of the SDM-derived response matrix, where ${\bf Q}\in {\mathbb C}^{N\times L}$ contains orthonormal columns and ${\bf \Delta}\in {\mathbb C}^{L\times L}$ is nonsingular.
Under an accurate SDM and the large array condition with $N\rightarrow\infty$, we have}
\vspace{-0.5em}\begin{equation}\label{eq:cm_projector}
{\bf Q} {\bf Q}^{\Htran} \rightarrow {\bf U} {\bf U}^{\Htran}.
\vspace{-0.25em}\end{equation}
\end{thm}
\begin{IEEEproof}
The details are relegated to Appendix \ref{appe:decomposition}.
\end{IEEEproof}
Then, the CM-RSLS estimator can be obtained as follows.
\begin{coro}\label{coro:CMLS}
A CM-RSLS estimator of the NLoS component ${\bf h}^{\rm N}$ with $\{{\hat{\bf p}}_l:\forall l\}$ provided by the ${\cal M}_{\rm SDM}$ is
\vspace{-0.5em}\begin{equation}\label{eq:cm_rsls_estimate}
{\widehat{\bf h}}^{\rm N}_{\rm CM-RSLS} = \frac{{\bf Q} {\bf Q}^{\Htran}}{\sqrt{\rho}}{\bf y}^{\rm N}
\vspace{-0.25em}\end{equation}
Under the accurate SDM condition in Theorem~\ref{theo:decomposition}, this estimate approaches the RSLS estimate ${\widehat{\bf h}}^{\rm N}_{\rm RSLS}$ with perfect knowledge of ${\bf R}^{\rm N}$ in \eqref{eq:rsls estimate} as $N \rightarrow \infty$.
\end{coro}
\begin{IEEEproof}
The proof follows a similar approach as in \cite{demir2022channel} but with the subspace spanned by ${\bf Q}$.
\end{IEEEproof}
{It is worth noting that the CM-RSLS estimator constructs \({\bf Q}\) from only the \(3L\) SDM-provided scatterer coordinate entries through a compact QR decomposition with complexity \({\cal O}(NL^2)\).
When the dominant scatterer responses are well represented by the map-derived subspace, the CM-RSLS estimator can approach the accuracy of the RSLS estimator in \eqref{eq:rsls estimate} with substantially lower overhead and computational complexity.
Scatterer location errors, missing or spurious scatterers, and unrecorded environmental changes such as temporary blockage mainly affect performance by causing subspace mismatch, which may lead to projection loss or additional projected noise. Until the SDM is refreshed, LS can serve as a fallback without using the map.}

\section{An Extension to the Multi-UE Scenarios}\label{sec:NLoS estimation MU}

{UM-MIMO systems are inherently designed to serve multiple UEs simultaneously, where pilot reuse can cause pilot contamination and degrade channel estimation accuracy.
Following the sequential framework, we first estimate and subtract the LoS component of each UE using the single-UE method in Section~\ref{sec:LoS estimation}, and then extend the proposed SA-RSLS and CM-RSLS estimators to estimate the remaining NLoS components in multi-UE scenarios.
Specifically, we develop the multi-UE SA-RSLS (MUSA-RSLS) and multi-UE CM-RSLS (MUCM-RSLS) estimators, analyze their robustness to pilot contamination, and design a pilot assignment scheme that exploits the spatial distribution characteristics of scatterers among UEs.}

\begin{figure}[t!]
\centering
\includegraphics[scale=1]{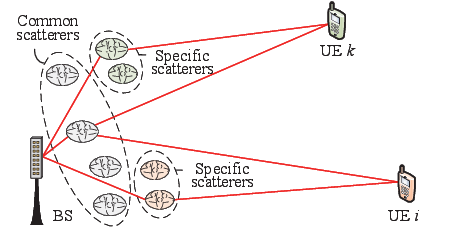}
\vspace{0em}
\caption{Illustration of the shared and exclusive scatterers for multiple UEs. 
\label{fig:2_system}}\vspace{-1.5em}
\end{figure}

\vspace{-1.em}
\subsection{Channel Model Considering the Visibility Limitation}
We consider a multi-UE UM-MIMO system where a BS serves $K$ UEs with an MPA.
Let ${\bf q}_k$ and ${\hat{\bf q}}_k$ denote the true and estimated positions of UE $k$, respectively.
As illustrated in Fig.~\ref{fig:2_system}, the scatterers are categorized into common and specific ones based on their contributions to the UEs' channel multipaths.
This classification is rooted in the {\it visibility limitation} of scatterers, which is determined by the LoS conditions between the scatterer, the BS, and each UE \cite{3GPPTR38901}.
A scatterer is only ``visible" to a UE if it can reflect that UE's transmitted signal to the BS, forming a valid NLoS multipath component.
This spatial non-stationary phenomenon of UM-MIMO coincides with the concept of the visibility region \cite{wang2021general}.
The {\it common scatterers} are visible to all $K$ UEs, typically located in open areas or central positions of the coverage region, such as public squares or building facades with wide reflection angles.
The {\it specific scatterers} for UE $k$ are only visible to UE $k$ and often include UE-specific obstacles or local structures, such as walls adjacent to a stationary UE or nearby foliage for a mobile UE.

We denote by ${\cal S}^{\rm co}$ and ${\cal S}^{\rm sp}_k$ the subsets of the common scatterers and the specific scatterers for UE $k$, respectively.
For $k \neq i$, we have ${\cal S}^{\rm sp}_k \cap {\cal S}^{\rm sp}_i = \emptyset$ to ensure exclusivity, as overlapping specific scatterers would inherently belong to the common category.
We let ${\bf h}_{k}^{\rm N} \in {\mathbb C}^{N}$ denote the NLoS channel vector between UE $k$ and the BS, which is given by
\vspace{-0.5em}\begin{equation}\label{eq:h2}
  {\bf h}_k^{\rm N} = \underbrace{\sum\nolimits_{l\in{\cal S}^{\rm co}} g_{kl} {\bf b}({\varphi}_l,{\theta}_l,{ r}_l)}_{\text{Common Component}} + \underbrace{\sum\nolimits_{j\in{\cal S}_k^{\rm sp}} g_{kj} {\bf b}({\varphi}_j,{\theta}_j,{ r}_j)}_{\text{Specific Component}}
\vspace{-0.25em}\end{equation}
by rewriting the expression in \eqref{eq:h1} as a superposition of the multipaths contributed by common and specific scatterers, where $g_{kl} \sim {\cal N}_{\mathbb C}(0, \beta_{kl}^{\rm N})$ with $\sum_{l=1}^L \beta_{kl}^{\rm N} = \beta_{k}^{\rm N}$ and $\beta_{k}^{\rm N}$ is the NLoS LSF coefficient of ${\bf h}_k^{\rm N}$.
A key structural property of \eqref{eq:h2} is the subspace orthogonality between specific components of different UEs, i.e., ${\bf b}^{\Htran}({\varphi}_j,{\theta}_j,{ r}_j){\bf b}({\varphi}_\ell,{\theta}_\ell,{ r}_\ell)\approx 0$, $\forall j\in{\cal S}_k^{\rm sp}$ and $\forall \ell \in{\cal S}_i^{\rm sp}$ with any $k \ne i$.
This orthogonality stems from the distinct spatial locations of specific scatterers, which result in non-overlapping AoA clusters for different UEs. In contrast, the array response vectors of common scatterers are identical across UEs, leading to overlapping subspace components between UEs.

For the convenience of the subsequent estimation, we rewrite the channel vector in \eqref{eq:h2} in a low-dimensional subspace form as
\vspace{-0.5em}\begin{equation}\label{eq:h3}
{\bf h}_k^{\rm N} = {\bf U}_k {\bf \Lambda}_k^{\frac{1}{2}} {\bf e}_k = {\bf U}_k {\bf g}_k  = {\bf U}^{\rm co} {\bf g}_k^{\rm co} + {\bf U}_k^{\rm sp} {\bf g}_k^{\rm sp}
\vspace{-0.25em}\end{equation}
by performing the Karhunen-Lo{\`e}ve expansion on ${\bf h}_k^{\rm N}$, where ${\bf g}_k = {\bf \Lambda}_k^{\frac{1}{2}} {\bf e}_k \sim {\cal N}_{\mathbb C}({\bf 0},{\bf \Lambda}_k)$ with ${\bf e}_k \sim {\cal N}_{\mathbb C}({\bf 0},{\bf I}_{r_k})$.
Matrices ${\bf U}_{k} \in {\mathbb C}^{N\times r_k}$ and ${\bf \Lambda}_{k} \in {\mathbb C}^{r_k \times r_k}$ are obtained by performing compact EVD on ${\bf R}_k^{\rm N}$ as ${\bf R}_k^{\rm N} = {\bf U}_{k} {\bf \Lambda}_{k} {\bf U}_{k}^{\Htran}$, where ${\bf R}_k^{\rm N} = {\mathbb E}\{{\bf h}_k^{\rm N}({\bf h}_k^{\rm N})^{\Htran}\}$ is the spatial correlation matrix of ${\bf h}_k^{\rm N}$ and $r_k={\rm Rank}({\bf R}_k^{\rm N})\ll  N$ is the rank of ${\bf R}_k^{\rm N}$.
Matrices ${\bf U}^{\rm co}\in {\mathbb C}^{N\times r^{\rm co}}$ and ${\bf U}_k^{\rm sp}\in {\mathbb C}^{N\times r_k^{\rm sp}}$ are the signal subspace matrices corresponding to the common and specific scatterers, respectively.
Vectors ${\bf g}_k^{\rm co} \in {\mathbb C}^{r^{\rm co}}$ and ${\bf g}_k^{\rm sp} \in {\mathbb C}^{r_k^{\rm sp}}$ are the low-dimensional coefficient vectors with ${\mathbb E}\{{\bf g}_k^{\rm co} ({\bf g}_k^{\rm co})^{\Htran}\} = {\bf \Lambda}_k^{\rm co}$ and ${\mathbb E}\{{\bf g}_k^{\rm sp} ({\bf g}_k^{\rm sp})^{\Htran}\} = {\bf \Lambda}_k^{\rm sp}$, respectively.
Here, $r^{\rm co}={\rm Rank}({\bf U}^{\rm co})$ and $r_k^{\rm sp}={\rm Rank}({\bf U}_k^{\rm sp})$ satisfying $r^{\rm co} + r_k^{\rm sp} = r_k$.
Moreover, ${\rm blkdiag}({\bf \Lambda}_k^{\rm co},{\bf \Lambda}_k^{\rm sp}) = {\bf \Lambda}_k$.

\vspace{0em}
\subsection{Multi-UE RSLS and CM-RSLS Channel Estimation}

During the pilot transmission phase, $K$ UEs are first assigned their pilot sequences from an orthogonal pilot set of cardinality $\tau_{\rm p} \ll K$.
Then, the received pilot signal at the BS in \eqref{eq:received signal at AP} becomes \cite[Sec.3]{bjornson2017massive}
\vspace{-0.5em}\begin{equation}\label{eq:received signal at AP MU}
{\bf y}_{k}^{\rm N} = \sqrt{\rho} \sum\nolimits_{i\in{\cal P}_{t_k}} {\bf h}_{i}^{\rm N} + {\bf n}_k\in {\mathbb C}^{N}
\vspace{-0.25em}\end{equation}
due to the pilot reuse among the UEs in ${\cal P}_{t_k}$, where $t_k$ is the index of the pilot assigned to UE $k$, ${\cal P}_{t_k}$ is the set of UEs sharing pilot $t_k$, and ${\bf n}_k \sim {\cal N}_{\mathbb C}({\bf 0}, {\bf I}_N)$ is the receiver noise.
{It is worth noting that \eqref{eq:received signal at AP MU} uses accurate LoS cancellation as a reference case.
With imperfect cancellation, ${\bf y}_{k}^{\rm N}$ additionally contains $\sqrt{\rho}\sum\nolimits_{i\in{\cal P}_{t_k}}({\bf h}_{i}^{\rm L}- {\widehat{\bf h}}_{i}^{\rm L})$, whose projection contributes residual interference to the multi-UE estimates.}
The BS projects ${\bf y}_{k}^{\rm N}$ onto ${\bf U}_{k}$ as ${\bf y}_{k}^{\rm N,p} = {\bf U}_{k}^{\Htran}{\bf y}_{k}^{\rm N}$.
By exploiting \eqref{eq:h3} and the subspace orthogonality, we have
\vspace{-0.5em}\begin{equation}\label{eq:received signal at AP MU1}
{\bf y}_{k}^{\rm N,p} = \sqrt{\rho} {\bf g}_{k} + \sqrt{\rho} \sum\nolimits_{i\in{\cal P}_{t_k},i\ne k} {\bf U}_{k}^{\Htran}{\bf U}^{\rm co}{\bf g}_{i}^{\rm co} + {\bf U}_{k}^{\Htran}{\bf n}_k\in {\mathbb C}^{r_k}
\vspace{-0.25em}\end{equation}
where the pilot contamination is caused by the residual common components of the pilot-sharing UEs.
Sequentially, a multi-UE RSLS (MU-RSLS) estimate is obtained by applying the LS estimator to minimize $\|{\bf y}_{k}^{\rm N,p} - \sqrt{\rho} {\bf g}_{k}\|^2$ and then projecting the estimate back into ${\bf U}_{k}$.
\begin{lemm}\label{coro:MURSLS}
A MU-RSLS estimator of the NLoS component ${\bf h}_k^{\rm N}$ considering pilot contamination is
\vspace{-0.5em}\begin{equation}\label{eq:mursls estimate}
\begin{aligned}
&{\widehat{\bf h}}^{\rm N}_{{\rm MU-RSLS},k} \!=\! \frac{{\bf U}_k {\bf U}_k^{\Htran}}{\sqrt{\rho}}{\bf y}_{k}^{\rm N}\\
&\!=\! {\bf U}_k{\bf g}_{k} + \sum\nolimits_{i\in{\cal P}_{t_k},i\ne k} {\bf U}^{\rm co} {\bf g}_{i}^{\rm co} +\frac{{\bf U}_k {\bf U}_k^{\Htran}}{\sqrt{\rho}}{\bf n}_k.
\end{aligned}
\vspace{-0.25em}\end{equation}
\end{lemm}
\begin{IEEEproof}
The proof follows a similar approach as in \cite{demir2022channel} and thus is omitted here due to the limited space.
\end{IEEEproof}

Similarly, we extend SA-RSLS to multi-UE scenarios as the MUSA-RSLS estimator.
\begin{coro}\label{coro:MUSARSLS}
A MUSA-RSLS estimator of the NLoS component ${\bf h}_k^{\rm N}$ considering pilot contamination is
\vspace{-0.5em}\begin{equation}\label{eq:musarsls estimate}
{\widehat{\bf h}}^{\rm N}_{{\rm MUSA-RSLS},k} = \frac{{\bf U}_{{\rm s},k} {\bf U}_{{\rm s},k}^{\Htran}}{\sqrt{\rho}}{\bf y}_k^{\rm N}
\vspace{-0.25em}\end{equation}
where ${\bf U}_{{\rm s},k} \in {\mathbb C}^{N\times {\tilde r}_k}$ is the sketch-based subspace matrix specific to UE $k$ and ${\tilde r}_k$ is the sketch size of UE $k$.
Matrix ${\bf U}_{{\rm s},k}$ is obtained using the procedure for ${\bf U}_{\rm s}$ described in Section~\ref{subsec:sa rsls}.
\end{coro}
\begin{IEEEproof}
The proof follows a similar approach as in Section \ref{subsec:sa rsls} but with the subspace for UE $k$.
\end{IEEEproof}

{Since ${\bf U}_{{\rm s},k}$ is constructed from ${\bf R}_k^{\rm N}$, pilot contamination does not alter the sketch basis under the assumed covariance knowledge, but affects the estimate through the pilot-sharing components retained after projection.}

Inspired by CM-RSLS, we propose a MUCM-RSLS scheme that uses the SDM to provide favorable estimation accuracy with low overhead and complexity.
In this case, the SDM is expected to provide not only the prior positions of scatterers $\{{\hat{\bf p}}_l:\forall l\}$ but also the scatterer category $\{{\cal S}^{\rm co},{\cal S}^{\rm sp}_k:\forall k\}$, based on the positions of UEs $\{{\hat{\bf q}}_k:\forall k\}$ and the environment database $\cal E$.
Mathematically, we express this SDM as
\vspace{-0.5em}\begin{equation}
{\cal M}_{\rm SDM}:\
\{{\hat{\bf p}}_l,{\cal S}^{\rm co},{\cal S}^{\rm sp}_k:\forall l,k\} = f_{\rm SDM}(\{{\hat{\bf q}}_k:\forall k\},{\cal E}),
\vspace{-0.25em}\end{equation}
where the UE positions are obtained by the LoS estimator in Section~\ref{sec:LoS estimation} using the configured MPA.
\begin{coro}\label{coro:MUCMRSLS}
A MUCM-RSLS estimator of the NLoS component ${\bf h}_k^{\rm N}$ with $\{{\hat{\bf p}}_l,{\cal S}^{\rm co},{\cal S}^{\rm sp}_k:\forall l,k\}$ provided by the ${\cal M}_{\rm SDM}$ is
\vspace{-0.5em}\begin{equation}\label{eq:mucmrsls estimate}
{\widehat{\bf h}}^{\rm N}_{{\rm MUCM-RSLS},k} = \frac{{\bf Q}_k {\bf Q}_k^{\Htran}}{\sqrt{\rho}}{\bf y}_k^{\rm N}
\vspace{-0.25em}\end{equation}
where ${\bf Q}_k\in {\mathbb C}^{N\times L_k}$ is obtained by performing the compact QRD on ${\hat{\bf U}}_k = [[{\bf b}({\hat\varphi}_l,{\hat\theta}_l,{\hat r}_l)]_{l \in {\cal S}^{\rm co}},[{\bf b}({\hat\varphi}_\ell,{\hat\theta}_\ell,{\hat r}_\ell)]_{l \in {\cal S}_k^{\rm sp}}]\in {\mathbb C}^{N\times L_k}$ with $L_k = |{\cal S}^{\rm co}|+|{\cal S}_k^{\rm sp}|$.

With accurate SDM positions and common/specific classification, the proposed MUCM-RSLS estimate ${\widehat{\bf h}}^{\rm N}_{{\rm MUCM-RSLS},k}$ approaches the MU-RSLS estimate ${\widehat{\bf h}}^{\rm N}_{{\rm MU-RSLS},k}$ with perfect knowledge of ${\bf R}_k^{\rm N}$ in \eqref{eq:mursls estimate} as $N \rightarrow \infty$.
\end{coro}
\begin{IEEEproof}
The proof follows a similar approach as in Appendix \ref{appe:decomposition} and \cite{demir2022channel} but with the subspace for UE $k$.
\end{IEEEproof}

The proposed MUCM-RSLS estimator has the same overhead and complexity order as CM-RSLS, which are substantially lower than those of MU-RSLS in multi-UE UM-MIMO scenarios.
{Errors in ${\hat{\bf q}}_k$ or sudden blockages can alter the predicted visibility, while scatterer location errors perturb the resulting ${\bf Q}_k$.
When a visibility change makes the common/specific classification outdated, the affected MUCM-RSLS subspace and the structural similarities used for pilot assignment should be refreshed before reuse.}

\vspace{-1em}
\subsection{NMSE and Structural Similarity-based Pilot Assignment}

According to Corollary \ref{coro:MUCMRSLS}, the NMSE of the MUCM-RSLS estimator for UE $k$ approaches that of the MU-RSLS estimator, which is derived as
\vspace{-0.5em}\begin{equation}\label{eq:nmse mursls}
\begin{aligned}
{\sf NMSE}_{{\rm MU-RSLS},k} &= \frac{{\mathbb E}\{\| {{\bf h}}^{\rm N}_{k} -  {\widehat{\bf h}}^{\rm N}_{{\rm MU-RSLS},k}\|^2\}}{{\mathbb E}\{\| {{\bf h}}^{\rm N}_{k}\|^2\}}\\
&= \frac{\sum\nolimits_{i\in{\cal P}_{t_k},i\ne k} {\rm tr}({\bf \Lambda}_i^{\rm co})+\frac{r_k}{\rho}}{{\rm tr}({\bf \Lambda}_k)}
\end{aligned}
\vspace{-0.25em}\end{equation}
after some algebraic manipulations.
We can see from \eqref{eq:nmse mursls} that the NMSE is primarily determined by the power of common subspace components from the pilot-sharing UEs and
\vspace{-0.5em}\begin{equation}
\frac{r_k}{\rho{\rm tr}({\bf \Lambda}_k)} {\le} {\sf NMSE}_{{\rm MU-RSLS},k} {\le}  \frac{\sum\nolimits_{i\in{\cal P}_{t_k},i\ne k} {\rm tr}({\bf \Lambda}_i)+\frac{r_k}{\rho}}{{\rm tr}({\bf \Lambda}_k)}.
\vspace{-0.25em}\end{equation}
The equality in the first inequality holds when there is no pilot contamination, and the equality in the second inequality occurs in the case of no specific scatterers.
This implies that the subspace overlap of the pilot-sharing UEs should be as low as possible to suppress pilot contamination such that ${\sf NMSE}_{{\rm MU-RSLS},k} \rightarrow  \frac{r_k}{\rho{\rm tr}({\bf \Lambda}_k)}$.
To characterize the degree of subspace overlap between the UEs, we define a metric referred to as {\it structural similarity} as
\vspace{-0.5em}\begin{equation}\label{eq:similarity}
\rho_{ki} = \frac{|{\rm tr}({\bf \Lambda}_k^{\rm co}{\bf \Lambda}_i^{\rm co})|}{\sqrt{{\rm tr}({\bf \Lambda}_k)+{\rm tr}({\bf \Lambda}_i)}},\ \forall k,i.
\vspace{-0.25em}\end{equation}
A smaller $\rho_{ki}$ indicates weaker potential pilot contamination between UE $k$ and UE $i$.

Based on the structural similarities defined in \eqref{eq:similarity}, we propose a pilot assignment scheme that minimizes the total intra-group sum of structural similarities, i.e., $\sum\nolimits_{t=1}^{\tau_{\rm p}}\sum\nolimits_{k,i\in{\cal P}_{t},k\ne i} \rho_{ki}$.
The core idea is to partition the $K$ UEs into $\tau_p$ groups using K-means clustering and then assign one orthogonal pilot to each group.
The detailed procedure for K-means clustering is similar to that in \cite{chen2021structured} (and references therein) and is omitted here due to limited space.

\renewcommand\arraystretch{1.5}
\begin{table}[t!]
  \centering
  \fontsize{9}{9}\selectfont
  \caption{System Parameters. }\vspace{0em}
  \label{tab:paremeter}
    \begin{tabular}{|p{1.6cm}<{\centering}| p{4.5cm}<{\centering}| p{1.4cm}<{\centering}|}
    \hline
    \bf Parameters  & \bf Definitions  & \bf Values \cr\hline
     $L$, $N$ & Number of scatterers and antennas& $10$, $256$\cr\hline
     $f$ & Carrier frequency & 28 GHz\cr\hline
     $d$ & Minimum antenna distance& $\lambda/2$\cr\hline
     $\tau_{\rm c}$, $\tau_{\rm p}$ & Number of available channel uses and pilot sequences & 200, 5\cr\hline
     $\rho$, $\kappa$ & Pilot SNR and Rician K-factor &10 dB, 10 \cr\hline
     ${\varphi}_{\min}$, ${\varphi}_{\max}$ & Azimuth angular range & $-30^\circ$, $30^\circ$\cr\hline
     ${\theta}_{\min}$, ${\theta}_{\max}$ & Elevation angular range & $-20^\circ$, $0^\circ$\cr\hline
    \end{tabular}
\vspace{-2.em}
\end{table}

\vspace{-0.5em}
\section{Results and Discussion}\label{sec:results}

In this section, we will evaluate our proposed channel map-based channel estimation schemes in Section \ref{sec:LoS estimation}, Section \ref{sec:NLoS estimation}, and Section \ref{sec:NLoS estimation MU}.
Unless otherwise specified, the adopted system parameters are given in Table \ref{tab:paremeter}.
The simulations are performed in MATLAB R2022b on a computer with Intel Core i7-9700 CPU @ 3.00 GHz and 16 GB of RAM.

For estimating the LoS component, we propose two channel map-based schemes, namely ``CM-MUSIC" and ``CM-ML".
More precisely, CM-MUSIC employs the MUSIC algorithm~\cite{schmidt1986multiple} and CM-ML performs a discrete grid search to maximize the log-likelihood in \eqref{eq:likelihood} globally.
Both CM-MUSIC and CM-ML use antenna positions determined by the proposed channel map-based PGA method.
For the benchmarks ``MUSIC" and ``ML" without channel maps, their antennas remain in the UPA configuration.
{Accordingly, the reported LoS performance gains represent the combined benefit of the UDM-provided coarse position and the MPA geometry configured using this information.
Moreover, the curve labeled ``FILB" denotes the local FILB derived from \eqref{e10} under the adopted model with fixed covariance.}

For estimating the NLoS component, our proposed schemes are marked as ``SA-RSLS" and ``CM-RSLS".
The former employs the randomized sketching algorithm for significant complexity reduction while maintaining almost the same estimation accuracy, and the latter provides further reduction with the help of the SDM.
Five benchmarks are considered, i.e., ``GA-RSLS", ``NF-RSLS", ``FF-RSLS", ``MMSE", and ``LS".
The first three employ the RSLS estimator using the genie-aided perfect ${\bf R}^{\rm N}$, the near-field approximation ${\bar{\bf R}}^{\rm N}$~\cite{demir2024spatial}, and a far-field approximation that ignores the distance variables~\cite{demir2022channel}, respectively.
The last two perform the classical MMSE and LS estimations, respectively.
When considering pilot contamination in multi-UE scenarios, we propose the ``MUCM-RSLS" and ``MUSA-RSLS" estimation schemes and a structural similarity-based pilot assignment, namely ``SS-CM".
Two pilot assignment benchmarks are considered.
The first exploits the same K-means clustering but utilizes the global similarities $\{{\hat\rho}_{ki} = {|{\rm tr}({\bf \Lambda}_k{\bf \Lambda}_i)|}/{\sqrt{{\rm tr}({\bf \Lambda}_k)+{\rm tr}({\bf \Lambda}_i)}}:\forall k,i \}$ in the objective, namely ``GS-CM".
The second ``Random" performs a pilot assignment randomly.

\begin{figure}[t!]
\centering
\includegraphics[width=\columnwidth]{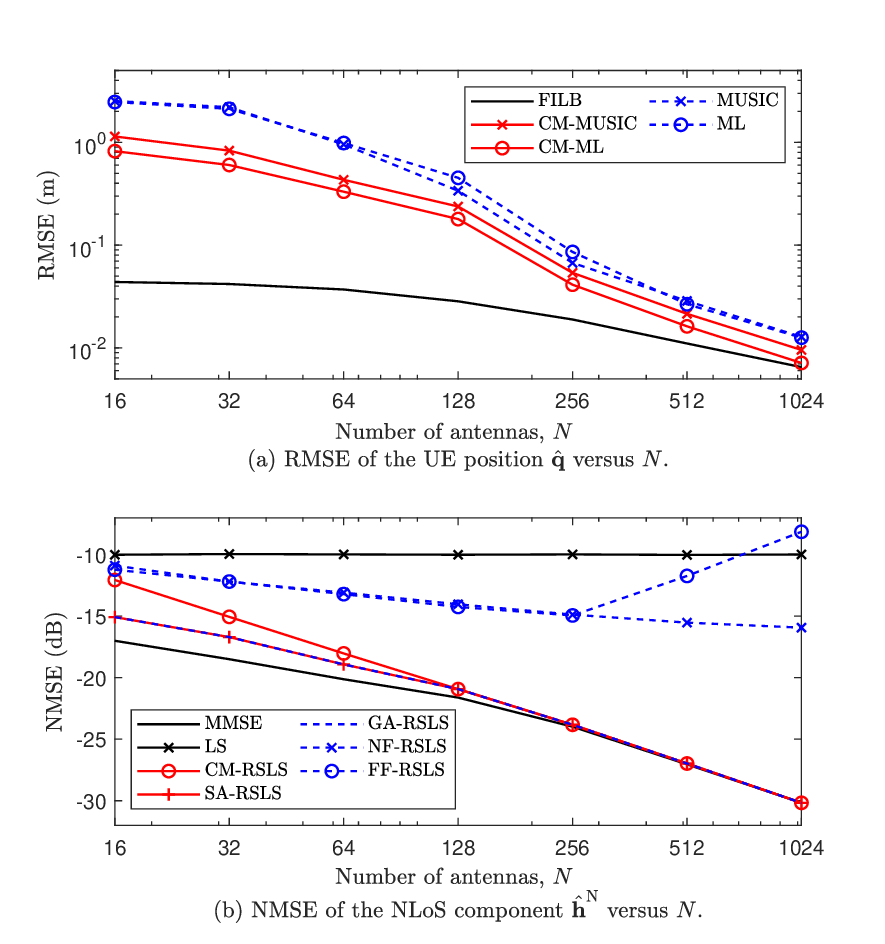}
\vspace{-2.em}
\caption{{(a) RMSE of the UE position ${\hat{\bf q}}$ and (b) NMSE of the NLoS component ${\hat{\bf h}}^{\rm N}$ versus the number of antennas $N$.}
\label{fig:3_NMSE_N}}
\vspace{-1.5em}
\end{figure}

\begin{figure}[t!]
\centering
\includegraphics[width=\columnwidth]{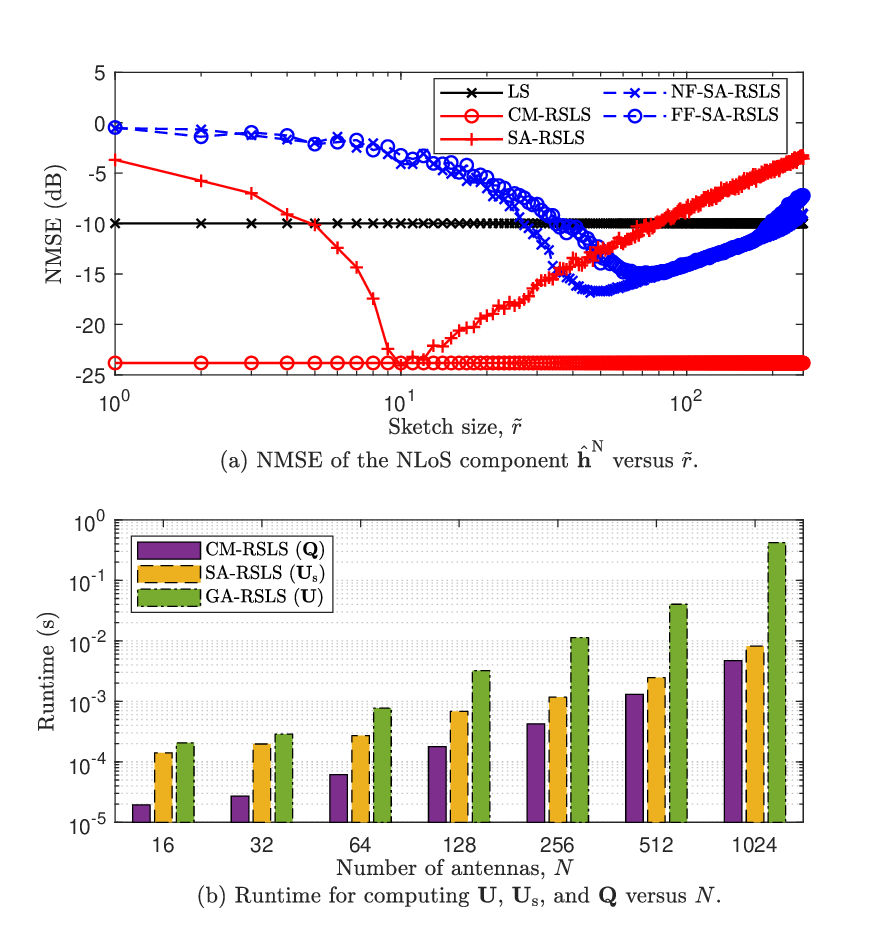}
\vspace{-2.em}
\caption{(a) NMSE of the NLoS component ${\hat{\bf h}}^{\rm N}$ versus the sketch size $\tilde r$ and (b) runtime for computing ${\bf U}$, ${{\bf U}_{\rm s}}$, and ${\bf Q}$ versus the number of antennas $N$.
\label{fig:5_NMSE_time_sketch}}
\vspace{-2.5em}
\end{figure}

We first quantify the UE localization accuracy versus antenna number $N$ in Fig.~\ref{fig:3_NMSE_N} (a), where the UDM provides a 10\%-error coarse UE position and the NLoS component is treated as noise.
As shown in Fig.~\ref{fig:3_NMSE_N} (a), the proposed CM-MUSIC and CM-ML outperform all benchmarks in RMSE across all $N$ with RMSE decreasing steadily and {approaching the local FILB under the adopted simulation model.
This advantage reflects the joint use of the UDM-provided coarse UE position and the MPA geometry configured through PGA, which together improve the FIM conditioning of LoS localization.}
To highlight the proposed NLoS estimation schemes, we assume that the UE position is known and that the SDM provides error-free scatterer positions in Fig.~\ref{fig:3_NMSE_N} (b).
We can observe that all schemes exhibit a decreasing NMSE with increasing $N$, while the proposed CM-RSLS and SA-RSLS outperform NF-RSLS, FF-RSLS, and LS by leveraging the scattering structure and low-rank property of NLoS channels.
SA-RSLS consistently matches GA-RSLS closely, while CM-RSLS converges to GA-RSLS as $N$ grows, as expected in Corollary \ref{coro:CMLS}.
For large $N$, all three approaches converge toward the MMSE benchmark.
FF-RSLS deteriorates for $N \geq 256$ because the enlarged aperture invalidates its far-field plane-wave model.

\begin{figure}[t!]
\centering
\includegraphics[width=\columnwidth]{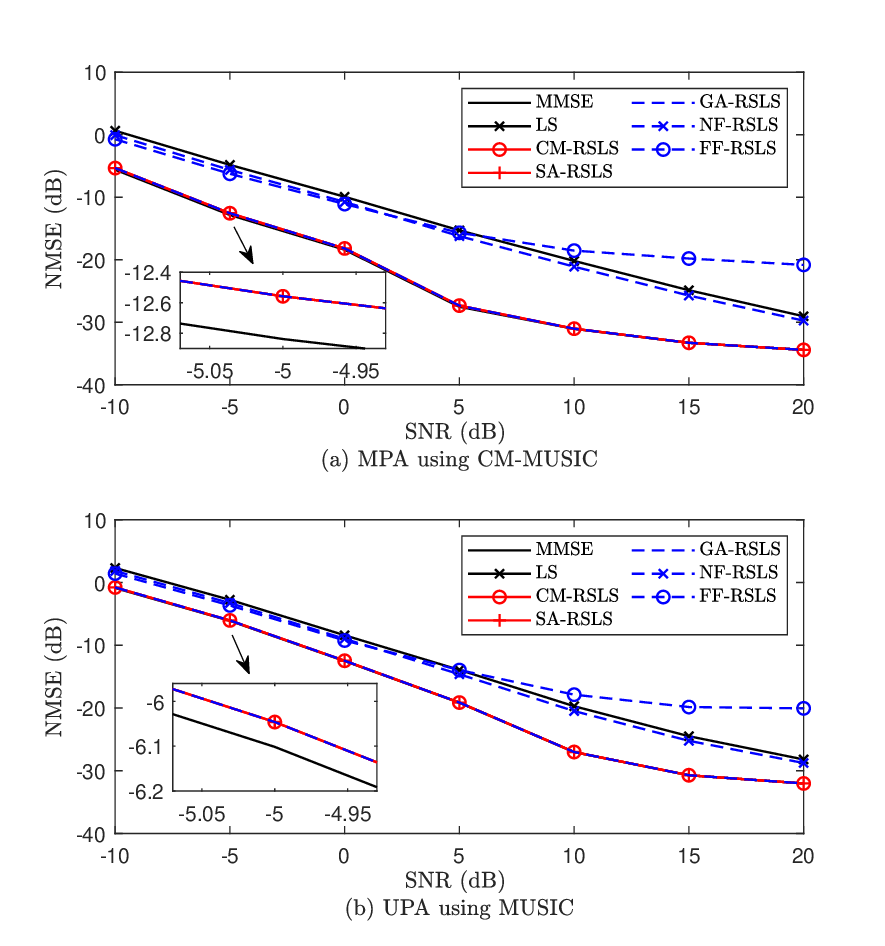}
\vspace{-2.em}
\caption{NMSE of the full channel ${\hat{\bf h}}$ versus the SNR $\rho$ employing (a) a MPA with CM-MUSIC and (b) a UPA with MUSIC.
\label{fig:4_NMSE_SNR}}
\vspace{-2em}
\end{figure}

{Fig.~\ref{fig:5_NMSE_time_sketch} (a) examines the NMSE versus sketch size, where ``NF-SA-RSLS" and ``FF-SA-RSLS" apply sketching to NF-RSLS and FF-RSLS, respectively.
The SA-RSLS NMSE decreases as ${\tilde r}$ approaches the effective channel rank and closely matches GA-RSLS around this rank.
A smaller ${\tilde r}$ causes projection loss, whereas a larger one retains additional noise directions.
CM-RSLS achieves the lowest NMSE, while NF-SA-RSLS and FF-SA-RSLS are affected by mismatch in their underlying approximate subspaces.}
Fig.~\ref{fig:5_NMSE_time_sketch} (b) compares the runtime for subspace matrix extraction under different antenna array scales. GA-RSLS has the longest runtime due to full-dimensional EVD. SA-RSLS shortens the runtime via compact QRD with low-dimensional EVD, while CM-RSLS is the most efficient, as the SDM allows it to determine the subspace basis without EVD.
This aligns with the complexity analysis in Section~\ref{sec:NLoS estimation}.
Specifically, compared to GA-RSLS, our proposed SA-RSLS and CM-RSLS reduce the runtime by two and three orders of magnitude, respectively, at $N=1024$.
\begin{figure}[t!]
\centering
\includegraphics[width=\columnwidth]{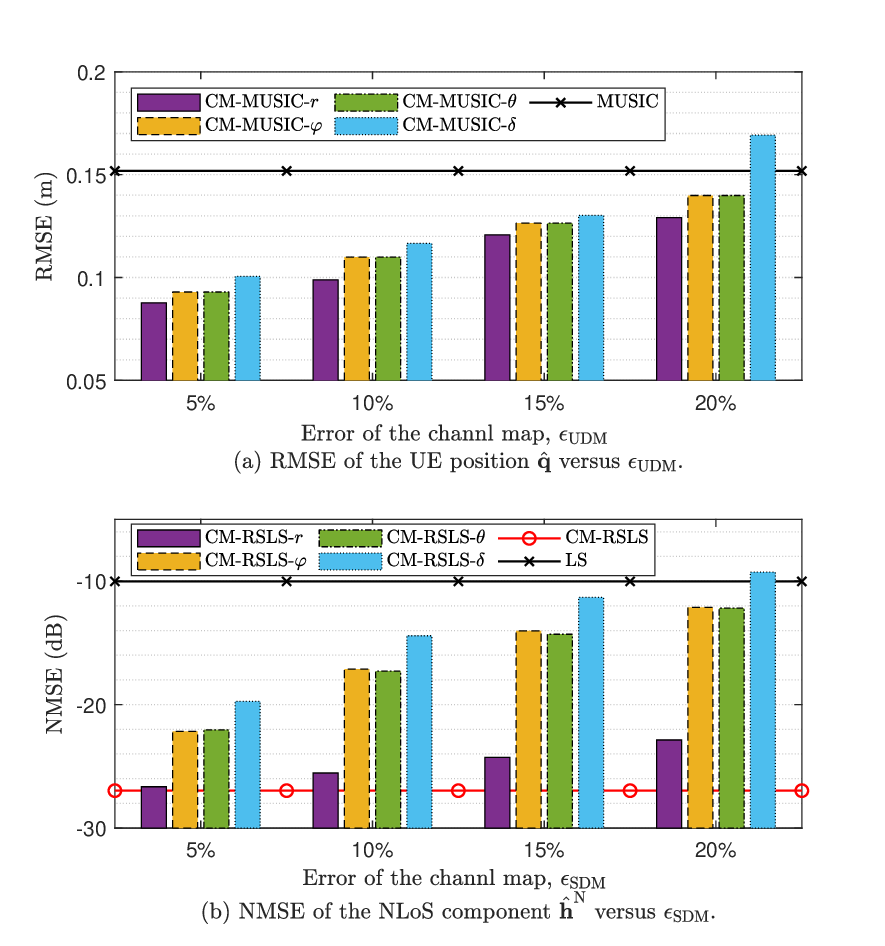}
\vspace{-2.em}
\caption{{(a) RMSE of the UE position ${\hat{\bf q}}$ versus channel map error $\epsilon_{\rm UDM}$ and (b) NMSE of the NLoS component ${\hat{\bf h}}^{\rm N}$ versus channel map error $\epsilon_{\rm SDM}$.}
\label{fig:6_NMSE_N_error}}
\vspace{-2em}
\end{figure}

In Fig.~\ref{fig:4_NMSE_SNR}, we evaluate our proposed LoS and NLoS estimation schemes versus SNR for the full channel.
Fig.~\ref{fig:4_NMSE_SNR} (a) employs an MPA with CM-MUSIC for LoS estimation, while Fig.~\ref{fig:4_NMSE_SNR} (b) uses a UPA with MUSIC.
As observed, all schemes exhibit decreasing NMSE with increasing SNR due to the reduced noise.
At high SNR, the NMSE decrease becomes gradual as the LoS localization error approaches its floor, limiting further improvement in full-channel estimation.
{Consistent with Fig.~\ref{fig:3_NMSE_N} (b), SA-RSLS and CM-RSLS outperform NF-RSLS, FF-RSLS, and LS while matching GA-RSLS closely.
In particular, SA-RSLS remains close to GA-RSLS from $-10$ to $20$ dB, showing little sketch-specific loss in the tested low-SNR region.}
Comparing Fig.~\ref{fig:4_NMSE_SNR} (a) and Fig.~\ref{fig:4_NMSE_SNR} (b), all schemes in (a) achieve lower NMSE than their counterparts in (b), consistent with the more accurate LoS estimation of CM-MUSIC observed in Fig.~\ref{fig:3_NMSE_N} (a).
{Under the considered setup, these results illustrate the benefit of the proposed joint design, in which the UDM-provided coarse position guides the MPA configuration for LoS estimation.
Thus, the full channel results include the propagation of LoS estimation errors through the single-UE sequential processing.}

Fig.~\ref{fig:6_NMSE_N_error} analyzes the impact of channel map errors on the RMSE of the proposed CM-MUSIC and the NMSE of the proposed CM-RSLS.
Suffixes $-\delta$, $-\varphi$, $-\theta$, $-r$ denote the channel errors in overall position, azimuth, elevation, and distance, respectively.
The estimation errors of both channel map-based schemes increase with channel map errors.
Notably, distance errors ($-r$) have the smallest impact; azimuth ($-\varphi$) and elevation ($-\theta$) errors exert nearly equal but milder effects than overall position errors ($-\delta$).
Even at 10\% overall channel map error, CM-MUSIC maintains a 22.2\% RMSE reduction over MUSIC, and CM-RSLS maintains a 4.5 dB NMSE edge over LS. This confirms the robustness of the proposed channel map-based schemes to moderate map inaccuracies.
{Together with Fig.~\ref{fig:4_NMSE_SNR}, these results separately characterize the effects of LoS estimation errors and SDM coordinate perturbations on the downstream estimates under the considered settings.}

\begin{figure}[t!]
\centering
\includegraphics[width=\columnwidth]{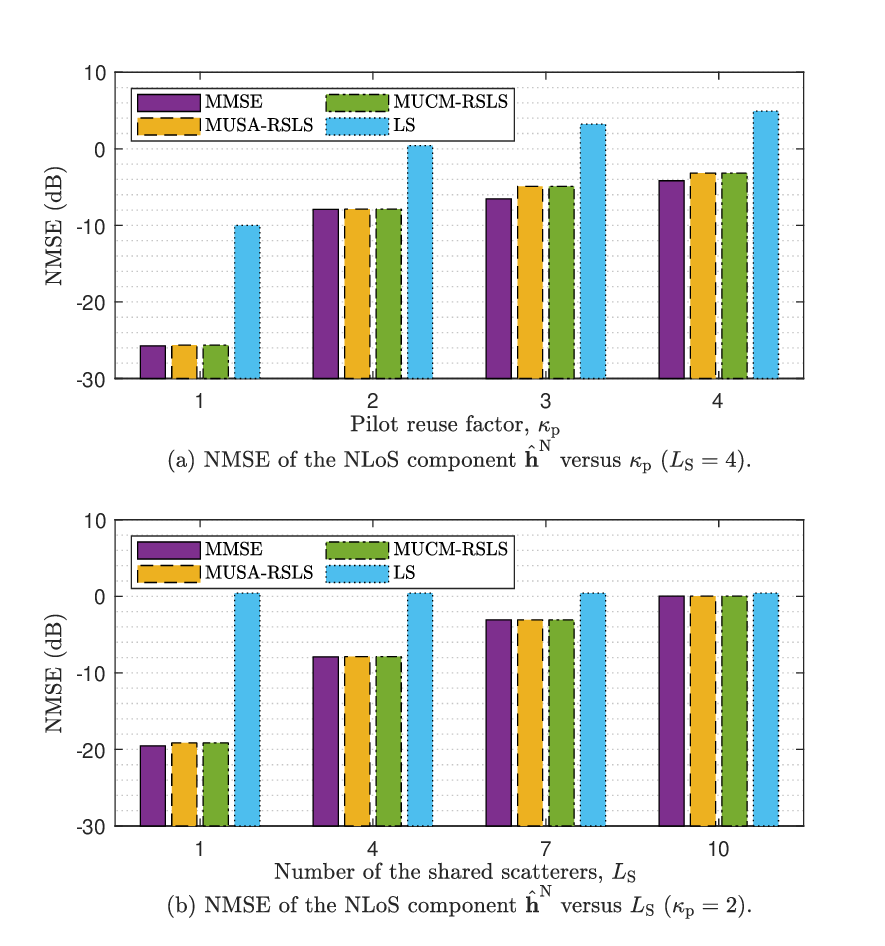}
\vspace{-2.em}
\caption{NMSE of the NLoS component ${\hat{\bf h}}^{\rm N}$ versus (a) the pilot reuse factor $\kappa_{\rm p}$ and (b) the number of shared scatterers $L_{\rm S}$ in multi-UE scenarios.
\label{fig:7_NMSE_MU}}
\vspace{-2em}
\end{figure}

In Fig.~\ref{fig:7_NMSE_MU}, we evaluate the NMSE of the proposed MUCM-RSLS and MUSA-RSLS in multi-UE scenarios with pilot reuse.
Fig.~\ref{fig:7_NMSE_MU} (a) demonstrates the impact of the pilot reuse factor $\kappa_{\rm p}$ on NMSE, with the number of shared scatterers $L_{\rm S} = 4$.
As shown, the NMSE of all schemes increases with $\kappa_{\rm p}$, indicating that stronger pilot reuse causes more severe pilot contamination and degrades estimation accuracy.
{MUCM-RSLS and MUSA-RSLS perform nearly identically as $\kappa_{\rm p}$ increases, indicating that the degradation is governed mainly by pilot reuse and common subspace overlap rather than randomized subspace acquisition.}
Fig.~\ref{fig:7_NMSE_MU} (b) examines the effect of the number of shared scatterers $L_{\rm S}$ on NMSE, with $\kappa_{\rm p} = 2$.
For all schemes except LS, NMSE increases with $L_{\rm S}$ since more shared scatterers expand the common subspace, amplifying pilot contamination from pilot-sharing UEs.
At $\kappa_{\rm p}=2$ and $N=256$, MUCM-RSLS and MUSA-RSLS remain close to MMSE as $L_{\rm S}$ increases, confirming effective pilot-contamination mitigation.

Finally, we evaluate the uplink spectral efficiency (SE) of the proposed SS-CM pilot allocation scheme in multi-UE scenarios with pilot contamination in Fig.~\ref{fig:8_SE_MU}.
We employ the hardening bound in \cite{demir2022cell} to compute SE, using $(\tau_{\rm c} - \tau_{\rm p} )$ channel uses for uplink data reception and maximum ratio combining to highlight the impact of channel estimation on decoding.
As shown in both Fig.~\ref{fig:8_SE_MU} (a) with $\kappa_{\rm p}=2$ and Fig.~\ref{fig:8_SE_MU} (b) with $\kappa_{\rm p}=4$, the SS-CM pilot assignment scheme that leverages the structural similarities outperforms GS-CM that leverages the global similarities.
This is because global similarities fail to accurately quantify subspace overlap between UEs, leading to suboptimal pilot grouping.
Both SS-CM and GS-CM far exceed the baseline Random.
Notably, UEs employing MUCM-RSLS channel estimates for combining achieve much higher SE than those using LS estimates, as MUCM-RSLS provides more accurate channel information, thereby enhancing interference suppression during data decoding.
Comparing Fig.~\ref{fig:8_SE_MU} (a) and Fig.~\ref{fig:8_SE_MU} (b), the SE at $\kappa_{\rm p}=2$ is over twice that at $\kappa_{\rm p}=4$ since lower pilot reuse reduces contamination and improves channel estimation accuracy.

\begin{figure}[t!]
\centering
\includegraphics[width=\columnwidth]{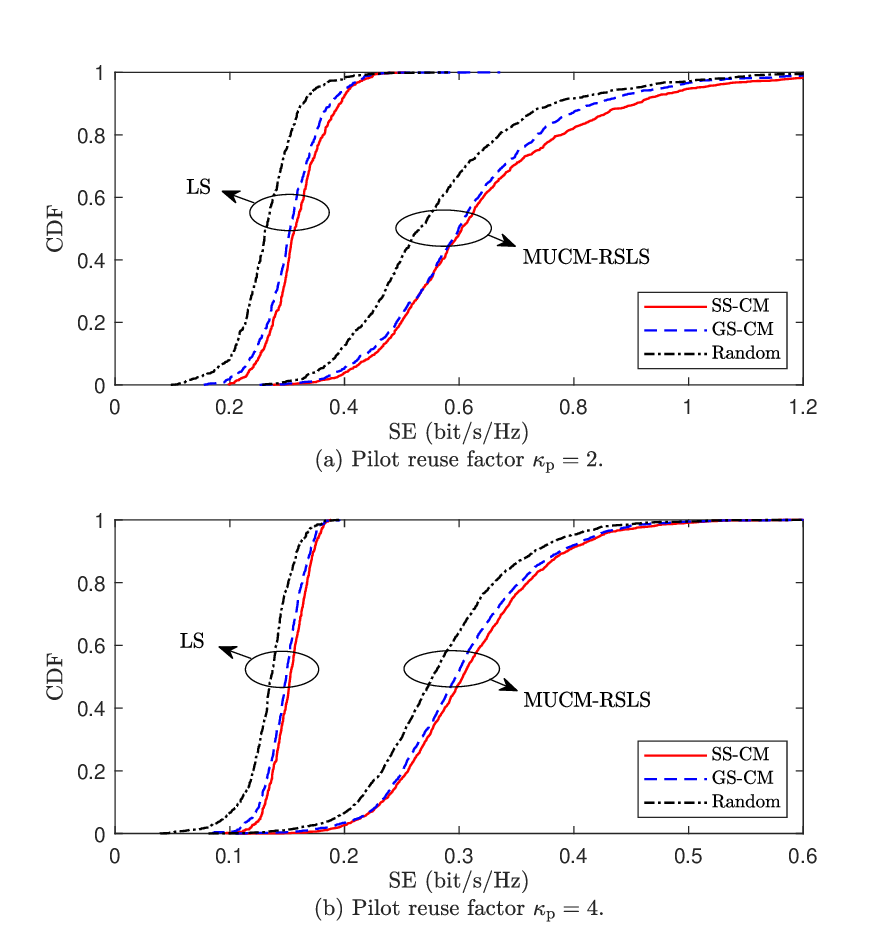}
\vspace{-2.em}
\caption{CDFs of SE per UE with pilot reuse factors (a) $\kappa_{\rm p} = 2$  and (b) $\kappa_{\rm p} = 4$ in multi-UE scenarios.
\label{fig:8_SE_MU}}
\vspace{-2em}
\end{figure}
\section{Conclusions}\label{sec:conclusion}

This paper has proposed a unified channel map-based near-field channel estimation framework for UM-MIMO systems with MPAs.
By integrating environment-aware channel maps with the geometric advantages of antenna mobility and the characteristics of near-field propagation, the framework has provided a comprehensive and coherent set of approaches for estimating LoS and NLoS components and for addressing multi-UE pilot contamination.
We have proposed a channel map-based LoS estimation method that combines UE position priors with Fisher information-guided antenna placement, enabling accurate LoS parameter acquisition under spherical-wave conditions.
Efficient NLoS estimation methods have also been proposed, leveraging low-rank subspace structure and scatterer-location information through the SA-RSLS and CM-RSLS estimators.
Moreover, we have extended the NLoS estimation framework to multi-UE scenarios by modeling scatterer visibility regions and designing the SS-CM pilot assignment strategy, which together enhance estimation robustness and improve SE.
Our simulation results have demonstrated consistent and significant gains in estimation accuracy, pilot contamination mitigation, and computational efficiency compared with conventional near-field and far-field benchmarks.
Overall, this work has shown that channel map-based technology provides a scalable and environment-adaptive foundation for high-accuracy channel acquisition in near-field UM-MIMO, highlighting its strong potential as an enabling technology for upcoming 6G wireless communications.

\begin{appendices}
\section{Proof of Lemma \ref{lemm:D_optimality}}\label{appe:D_optimality}

Let $\xi  = {{2{{\left| \alpha \right|}^2} \chi^2 }}/{{{\sigma ^2}}}$, and define $\mathbf{J} = \xi \sum_{n=1}^N \mathbf{J}_n$, where $\mathbf{J}_n = {{\left( {{\tilde{\bf q}} - {\mathbf{a}_n}} \right){{\left( {{\tilde{\bf q}} - {\mathbf{a}_n}} \right)}^\top}}}/{ {\left\| {{\tilde{\bf q}} - {\mathbf{a}_n}} \right\|^2}}$.
By linearity, differentiation with respect to the $n$-th antenna position affects only the $n$-th term, yielding
\vspace{-0.5em}\begin{equation}
\frac{\partial \mathbf{J}}{\partial [\mathbf{a}_{n}]_j} = \xi \frac{\partial \mathbf{J}_n}{\partial [\mathbf{a}_{n}]_j}.
\vspace{-0.25em}\end{equation}
Let $\mathbf{v}_n = {\tilde{\bf q}} - \mathbf{a}_n$, so that $\mathbf{J}_n = {\mathbf{v}_n \mathbf{v}_n^\top}/{\mathbf{v}_n^\top \mathbf{v}_n}$.
	Let $\mathbf{A} = \mathbf{v}_n \mathbf{v}_n^\top$ and $b = \mathbf{v}_n^\top \mathbf{v}_n = \left\| {{\mathbf{v}_n}} \right\|^2$.
	With ${\hat{\bf e}}_j$ being the standard basis vector for the $j$-th coordinate, we get
	\vspace{-0.5em}\begin{equation}
		\frac{\partial \mathbf{v}_n}{\partial [\mathbf{a}_{n}]_j} = -{\hat{\bf e}}_j
	\vspace{-0.25em}\end{equation}
We now compute the derivatives of $\mathbf{A}$ and $b$, as
	\vspace{-0.5em}\begin{equation}
		\frac{\partial b}{\partial [\mathbf{a}_{n}]_j} = \frac{\partial (\mathbf{v}_n^\top \mathbf{v}_n)}{\partial [\mathbf{a}_{n}]_j} = 2 \mathbf{v}_n^\top \frac{\partial \mathbf{v}_n}{\partial [\mathbf{a}_{n}]_j} = 2 \mathbf{v}_n^\top (-{\hat{\bf e}}_j) = -2 [\mathbf{v}_{n}]_j
	\vspace{-0.25em}\end{equation}
	\vspace{-0.5em}\begin{equation}
	\begin{split}
		& \frac{\partial \mathbf{A}}{\partial [\mathbf{a}_{n}]_j} = \frac{\partial (\mathbf{v}_n \mathbf{v}_n^\top)}{\partial [\mathbf{a}_{n}]_j} = \left(\frac{\partial \mathbf{v}_n}{\partial [\mathbf{a}_{n}]_j}\right) \mathbf{v}_n^\top + \mathbf{v}_n \left(\frac{\partial \mathbf{v}_n}{\partial [\mathbf{a}_{n}]_j}\right)^\top \\
		& = -{\hat{\bf e}}_j \mathbf{v}_n^\top - \mathbf{v}_n {\hat{\bf e}}_j^\top.
	\end{split}
	\vspace{-0.25em}\end{equation}
	Applying the quotient rule, we have
	\vspace{-0.5em}\begin{equation}
		\frac{\partial \mathbf{J}_n}{\partial [\mathbf{a}_{n}]_j} = \frac{ (-{\hat{\bf e}}_j \mathbf{v}_n^\top - \mathbf{v}_n {\hat{\bf e}}_j^\top) \left\| {{\mathbf{v}_n}} \right\|^2 - (\mathbf{v}_n \mathbf{v}_n^\top)(-2[\mathbf{v}_{n}]_j) }{ (\left\| {{\mathbf{v}_n}} \right\|^2)^2 }.
	\vspace{-0.25em}\end{equation}
Substituting $\mathbf{v}_n = {\tilde{\bf q}} - \mathbf{a}_n$ back in, we get the complete expression for the FIM derivative in \eqref{e15}.

\section{Proof of Theorem \ref{theo:decomposition}}\label{appe:decomposition}

Under the accurate-SDM condition in Theorem~\ref{theo:decomposition}, ${\hat{\bf U}}={\bf U}^{\prime}={\bf Q}{\bf \Delta}$.
According to the array response vector definition in \eqref{eq:b}, for two distinct scatterers $l$ and $j$, the inner product of their array response vectors is given by
\vspace{-0.5em}\begin{equation}
{\bf b}^{\Htran}(\varphi_l,\theta_l,r_l){\bf b}(\varphi_j,\theta_j,r_j) = \sum_{n=1}^N \!e^{\jmath \phi_{n,lj}}
\vspace{-0.25em}\end{equation}
which is a sum of complex exponentials with phases
\vspace{-0.5em}\begin{equation}
\phi_{n,lj} = \chi\left[(d_{ln}-r_l)-(d_{jn}-r_j)\right],
\vspace{-0.25em}\end{equation}
where $d_{ln}=\|{\bf p}_l-{\bf a}_n\|$ and $d_{jn}=\|{\bf p}_j-{\bf a}_n\|$. These phases vary rapidly across the array and become dense and well distributed over $[0,2\pi]$ as $N$ increases.
The law of large numbers then implies that \cite[Lem. B.12]{bjornson2017massive}
\vspace{-0.5em}\begin{equation}
\lim_{N\rightarrow \infty} \frac{1}{N}\sum_{n=1}^N \!e^{\jmath \phi_{n,lj}}\!=\! \lim_{N\rightarrow \infty}  \frac{{\bf b}^{\Htran}(\varphi_l,\theta_l,r_l){\bf b}(\varphi_j,\theta_j,r_j)}{\|{\bf b}(\varphi_l,\theta_l,r_l)\| \|{\bf b}(\varphi_j,\theta_j,r_j)\|} \!=\! 0
\vspace{-0.25em}\end{equation}
Since $\|{\bf b}(\varphi_l,\theta_l,r_l)\| = \sqrt{N}$, $l=1,\ldots,L$, this indicates that the array response vectors become asymptotically orthogonal as $N \rightarrow \infty$.
In this case, both ${\bf U}^{\prime{\Htran}}{\bf U}^{\prime}$ and ${\bf \Delta}$ in the compact QR decomposition become diagonal matrices.
According to the definition of ${\bf h}^{\rm N}$ in \eqref{eq:h1}, we have
\vspace{-0.5em}\begin{equation}
\begin{aligned}
{\bf R}^{\rm N} &= \sum_{l=1}^L \beta^{\rm N}_l {\bf b}(\varphi_l,\theta_l,r_l){\bf b}^{\Htran}(\varphi_l,\theta_l,r_l)\\
&= {\bf U}^{\prime}\diag(\beta^{\rm N}_1,\ldots,\beta^{\rm N}_L){\bf U}^{\prime{\Htran}}\\
&= {\bf Q}{\bf \Delta}\diag(\beta^{\rm N}_1,\ldots,\beta^{\rm N}_L){\bf \Delta}^{\Htran}{\bf Q}^{\Htran}
\end{aligned}
\vspace{-0.25em}\end{equation}
where ${\bf \Delta}\diag(\beta^{\rm N}_1,\ldots,\beta^{\rm N}_L){\bf \Delta}^{\Htran} \rightarrow {\bf \Lambda}$ is a diagonal matrix and ${\bf Q} {\bf Q}^{\Htran} \rightarrow {\bf U} {\bf U}^{\Htran}$ when $N \rightarrow \infty$.
This completes the proof.
\end{appendices}



\end{document}